\documentclass[aop,preprint]{imsart}

\usepackage{amsthm,amsfonts,amsmath, amscd}
\usepackage[title]{appendix}

\usepackage{xcolor}

\newtheorem{thm}{THEOREM}[section]
\newtheorem{lm}[thm]{LEMMA}

\theoremstyle{definition}
\newtheorem{defi}[thm]{DEFINITION}
\theoremstyle{definition}
\newtheorem{remark}[thm]{Remark}

\arxiv{arXiv:1812.03874}

\startlocaldefs

\newcommand{\upchi}{\raise1pt\hbox{$\chi$}}
\newcommand{\R}{{\mathord{\mathbb R}}}

\newcommand{\N}{{\mathord{\mathbb N}}}

\renewcommand{\|}{{\Vert}}

\def\E{{\cal E}_{N,\alpha}}
\def\D{{\cal D}_{N,\alpha}}

\def\Do{{\cal D}_{N,1}}
\def\Dto{\widetilde{\cal D}_{N,1}}
\def\Em{{\cal E}_{N-1,\alpha}}
\def\dd{{\rm d}}
\def\ncht{\left(\begin{matrix} N\cr 2\cr \end{matrix}\right)}
\def\nmcht{\left(\begin{matrix} N-1\cr 2\cr \end{matrix}\right)}
\def\ST{{\mathcal S}_N}
\def\STE{{\mathcal S}_{N,E,p}}

\def\LN{L_{N,\alpha}}
\def\DN{\Delta_{N,\alpha}}

\endlocaldefs

\begin{document}

\begin{frontmatter}

\title{Spectral Gaps for Reversible Markov Processes with Chaotic Invariant Measures: The Kac Process with Hard Sphere Collisions in Three Dimensions}
\runtitle{Spectral Gaps and Chaos}

 \author{\fnms{Eric} \snm{Carlen}\corref{}\thanksref{t1,m1,m2}\ead[label=e1]{carlen@math.rutgers.edu}},
\author{\fnms{Maria} \snm{Carvalho}\thanksref{t2,t3,m1,m2}\ead[label=e2]{mcarvalh@math.rutgers.edu}}
\and
\author{\fnms{Michael} \snm{Loss}\thanksref{t4,m3} \ead[label=e3]{loss@math.gatech.edu}}

\thankstext{t1}{Partially supported by U.S. NSF grant DMS 1501007  and DMS 1764254}
\thankstext{t2}{Partially supported by F.C.T. grant  UID/MAT/04561/2013}
\thankstext{t3}{Partially supported by F.C.T. grant  SFRH/BSAB/113685/2015}
\thankstext{t4}{Partially supported by U.S. NSF grant DMS-1600560}
\runauthor{Carlen, Carvalho and Loss}

 \affiliation{Department of Mathematics, Rutgers University,\thanksmark{m1}\\  Department of Mathematics and CMAF-CIO, University of Lisbon,\thanksmark{m2}\\ 
 School of Mathematics, Georgia Institute of
Technology\thanksmark{m3}}

\address{Address of Eric Carlen\\Department of Mathematics, Rutgers University\\110 Frelinghuysen Road, Piscataway NJ 08854 USA\\ \printead{e1}}
\address{Address of Maria Carvalho\\DEpartment of Mathematics and CMAF-CIO\\University of Lisbon 1649-003 Lisbon, Portugal \\ \printead{e2}}
\address{Adress of Michael Loss\\School of Mathematics, Georgia Institute ofTechnology\\ Atlanta, GA 30332 USA \\ \printead{e3} }

\begin{abstract}
We develop a method for producing  estimates on the spectral gaps of reversible Markov jump processes with chaotic invariant measures, that is effective in the case of degenerate jump rates, and we apply it to prove the Kac conjecture for hard sphere collision in three dimensions. 
\end{abstract}



\end{frontmatter}

\section{Introduction} \label{intro}

In a seminal paper of 1956, Mark Kac \cite{K56}  introduced a family of continuous time reversible Markov jump processes on the sphere $S^{N-1}(\sqrt{N})$ of radius $\sqrt{N}$ 
 in $\R^N$. This family of processes, and its  generalizations, have drawn the attention of many researchers. Kac was motivated by a connection, in the large $N$ limit, to the non-linear Boltzmann equation.  
 The connection arises through a particular ``asymptotic independence'' property of  sequences 
 $\{{\rm d}\mu_N\}$, where ${\rm d}\mu_N$ is a probability measure on 
$S^{N-1}$.  This property  is possessed, in particular,  by the sequence $\{{\rm d}\sigma_N\}$  of uniform probability measures on $S^{N-1}(\sqrt{N})$. Let $\vec v = (v_1,\dots,v_N)$ denote 
a generic point on $S^{N-1}(\sqrt{N})$ of radius $\sqrt{N}$. Let $\phi$ be any bounded continuous function on $\R^k$ and 
${\rm d}\gamma = (2\pi)^{-1/2} e^{-v^2/2}{\rm d}v$ be the unit Gaussian probability measure on $\R$. As is well known, going back at least to Mehler \cite{M66}, 
$$
\lim_{N\to\infty}\int_{S^{N-1}(\sqrt{N})} \phi(v_1,\dots,v_k){\rm d}\sigma_N  = \int_{\R^k}\phi(v_1,\dots,v_k){\rm d}\gamma^{\otimes k}\ .$$
As long as one only looks at  coordinates belonging to a fixed, finite set, in the large $N$ limit, the coordinates in this  set are 
asymptotically  independent.   The main result of \cite{K56} concerned sequences of  probability measures
 $\{{\rm d}\mu_N\}$ on $S^{N-1}(\sqrt{N})$  with the property that, for some probability density $f$ on $\R$ with zero mean and unit variance,
 $$
\lim_{N\to\infty}\int_{S^{N-1}(\sqrt{N})} \phi(v_1,\dots,v_k){\rm d}\mu_N  = \int_{\R^k}\phi(v_1,\dots,v_k) \prod_{j=1}^k f(v_j){\rm d}v_j ,$$
in which case  the sequence $\{{\rm d}\mu_N\}$ was said by Kac to be  {\em $f(v){\rm d}v$ chaotic}. 
He proved  that chaoticity was propagated in time by solutions of the forward Kolmogrov equations associated to the Kac processes.  Moreover, if   
$\{{\rm d}\mu_N(t)\}$ is the sequence of laws at time $t$ starting from an $f(v){\rm d}v$ chaotic sequence,  
$\{{\rm d}\mu_N(t)\}$ is $f(t,v){\rm d}v$ chaotic where $f(t,v)$ is the solution of the {\em Kac- Boltzmann equation} with initial data $f(v)$. 
(The Kac Boltzmann equation is a simple model of the Boltzmann equation for a gas in one dimension.)  He also made a conjecture,  
that went unsolved for a long time, concerning the spectral gap of the generator of this  family of processes. 
Since the processes are reversible, their generators are self adjoint, and it is not hard to see that the null space is spanned by the constants.  Kac conjectured a gap 
$\Delta_N$ separating $0$ from the rest of the spectrum that is bounded below uniformly in $N$. That is, $\lim_{N\to \infty}\Delta_N >0$. 
This was finally proved by Janvresse in 2000 \cite{J01}, and shortly afterwards the exact value of $\Delta_N$ for all $N$ was determined in \cite{CCL00}. 

A few years after his original work, Kac returned to these problems \cite{K59}, but this time for a physically realistic model of a gas in three dimensions  undergoing ``hard sphere'' collisions that 
conserve energy and momentum.  As he showed, this physical model would have, through propagation of chaos, a direct connection to the actual Boltzmann equation for 
hard sphere collisions, and not only a toy model of it. However, in the physical model, the rates at which different pairs of molecules collide depend on their velocities:  
The rates are not bounded away from $0$,  and there is no bound from above that is uniform in $N$. It is much harder to estimate spectral gaps for the generators of jump 
processes with rates that are not bounded from below, and the lack of an upper bound that is uniform in $N$ makes it much harder to prove propagation of chaos. 

In this paper, we prove the Kac conjecture for the Kac model with hard sphere collisions in $\R^3$. We do so using a method that has three essential components. These are:

\smallskip
\noindent{\it (1)} The introduction of a {\em conjugate process}, in which at each step all but one of the velocities are updated. The rates in this process  are still not bounded below, but they depend only on the one velocity that is left fixed during the jump. There is also a simple connection between the spectral gaps of the original process and the conjugate process, and the central problem becomes the determination of the spectral gap for the conjugate process.

\smallskip
\noindent{\it (2)} Quantitative estimates on the chaoticity of the sequence of invariant measures: We prove and apply estimates quantitatively expressing the near independence of any finite set of coordinates for large $N$.

\smallskip
\noindent{\it (3)} A trial function decomposition:  We decompose any trial function $f$ for the spectral gap problem into 3 pieces, $f = s + g + h$ that are 
mutually orthogonal in the $L^2$ space for the invariant measure, and, due to quantitative chaos estimates, are nearly orthogonal with respect to the 
inner product given by the Dirichlet form of the conjugate process. Each of these pieces has a particular special structure that facilitates the proof of 
 estimates of the type we seek.

\smallskip

The first two components have been present in our work on Kac type models since our early papers \cite{CCL00,CCL03} on the models (as in 
\cite{K56}) with uniform jump rates, though in the early papers, the conjugate process is not considered explicitly as a process. However, the connection 
between its spectral gap and the spectral gap for the Kac process has been central to the approach from the beginning.  Work by two of us and Jeff 
Geronimo \cite{CGL} dealt with the quantitative chaos estimates needed for the three dimensional energy and momentum conserving collision 
considered here, but applied them to ``Maxwellian molecules'' models which, unlike to hard sphere model, has rates that are bounded below. 
There too, the approach yielded the exact value of the the spectral gap for a wide class of ``Maxwellian molecules'' models.

Finally in \cite{CCL14} we proved the Kac conjecture for a ``hard sphere'' model  with one dimensional velocities, and introduced a somewhat simpler version of 
component {\it (3)}, the trial function decomposition. In application to kinetic theory, as explained in \cite{CCL14}, the spectral gap in the symmetric sector, i.e., 
for functions that are invariant under permutations of coordinates, is especially important. It is this quantity that can be related to the spectral gap for the linearized 
Boltzmann equation, and one would like to have explicit estimates on this gap. Therefore, in \cite{CCL14} we worked hard to render all estimates as sharp 
and explicit as possible, and to treat only the symmetric sector for which fewer estimates were required. 

It was clear to us at the time we wrote \cite{CCL14} that we had a general method that would prove the existence of a spectral gap, uniformly in $N$, for the physical three dimensional hard sphere Kac model, and we announced this in several lectures. The result is quoted in reference 9 of \cite{MM13}, as a personal communication, and used in the  development of the quantitative treatment of propagation of chaos that is provided there.  After our 
paper  \cite{CCL14} appeared  with the details provided only for the symmetric sector and the one dimensional model, St\'ephane Mischler and Cl\'ement Mouhot  asked us several times to provide the details. This paper answers their request, and moreover, in the course of preparing this answer, it has provided a clearer picture of the how the method explained \cite{CCL14} in can be extended and applied to more complicated models, such as the main example treated here. 

The method to be explained here may be applied to a wide class of sequences of reversible Markov jump processes whose sequence of invariant measures satisfies certain ``quantitative chaos'' estimates that are specified here. 
The method is not at all restricted to the 
treatment of the symmetric sector, and perhaps had we explained the method in \cite{CCL14} without obscuring it behind the detail of so many 
explicit computations,  necessary for the precise quantitative estimates obtained there, this would have been clear some years ago.   

Therefore, in the present paper, we prove the Kac conjecture for hard sphere collisions in 
three dimensions without any symmetry condition in as simple a manner as possible  to provide a clear view of the method. To do this, we make use of constants 
$C$ that change from line to line but are independent of $N$ that are not explicitly evaluated here, but easily could be -- at the expense of 
more pages and less clarity.

In addition to the applications to quantitative propagation of chaos developed in \cite{MM13}, uniform bounds on the spectral gap are important in certain problems concerning the hydrodynamic limits of certain kinetic models, as explained in \cite{GKS12}. These authors considered a one dimensional model essentially equivalent to the one considered in \cite{CCL14}, and asked for the spectral gap. Sasada \cite{Sas15} provided the answer to the question they raised, noting that she could not simply apply the result of \cite{CCL14} as it applied to the symmetric sector only. This is true, but as shown here, the method used in \cite{CCL14} may  readily extended to answer a much broader range of questions. Much beautiful work has been done on the question of estimating spectral gaps for Kac type processes, and we refer to the papers of Sasada \cite{Sas15} and Caputo \cite{Cap04,Cap08}, in addition to our own papers cited here, for significant contributions. However, it is not clear to us that any of the other methods that have been developed for this class of models applies to the main example at hand which is considerably more complex than the models considered in most other work. 

\subsection{The Kac collision process} 
For $N\in \N$, $p\in \R^3$ and $E > |p|^2$,  let $\STE$ be the set
consisting of $N$--tuples  $\vec v = (v_1,\dots , v_N)$ of vectors $v_j$ in $\R^3$ with
${\displaystyle \frac1N\sum_{j=1}^N|v_j|^2 = E }$ and   ${\displaystyle \frac1N\sum_{j=1}^N v_j = p}$.
In what follows, a point $\vec v\in \STE$  specifies the velocities of a collection of $N$ particles with mass $2$, so that $E$ is 
the kinetic energy per particle, and $p$ is one-half  the momentum per particle. The Markov jump process  
 introduced by Mark Kac \cite{K59} describes a random binary collision process for the $N$ particles, in which   
the collisions  conserve both  energy and momentum,  and thus if the process starts on $\STE$, it will remain on $\STE$ for all time. 

Recall that a random variable $T$ with values in $(0,\infty)$ is {\em exponential with parameter} $\lambda$ in case
${\rm Pr}(T \geq t) =  e^{-\lambda t}$.  When the collision process begins, 
associated to each pair $(v_i,v_j)$,  $i<j$, is an exponential random variable  $T_{i,j}$ with parameter
\begin{equation}\label{jumprate}
\lambda_{i,j} =  N\ncht^{-1}|v_i - v_j|^{\alpha}\ ,
\end{equation}
where $0 \leq \alpha \leq 2$, and $\alpha =1$ is the case of main interest: As  explained in \cite{K59}, (\ref{jumprate}) is motivated by a connection between the Kac process and
the Boltzmann equation, and $\alpha =1$ corresponds to ``hard-sphere collisions''. 

$T_{i,j}$ represents the waiting time for  particles $i$ and $j$ to collide, and the set of these random times is taken to be independent. 
 The first collision occurs at time
\begin{equation}\label{alarm}
T = \min_{i<j}\{T_{i,j}\}\ .
\end{equation}
As is well known \cite{Fell}, the minimum of an independent set of  exponential random variables is itself exponential, and the parameter of the minimum is the sum of the parameters of the  random variables in the set.  In particular, if $\alpha =0$, $T$ is exponential with parameter $N$, and the expected waiting time for the
first collision of some pair to occur is $1/N$. 

At the time $T$, the pair $(i,j)$ furnishing the minimum collide:    The state of the process ``jumps'' from $ (v_1,\dots , v_N)$ to
$ (v_1,v_2,\dots,v_i^*,\dots,v_j^*,\dots,v_N)$,
where only $v_i$ and $v_j$ have changed. Since the process is conceived to model momentum and energy conserving collisions we require that
\begin{equation}\label{kinpos}
v_i^*+v_j^*  = v_i + v_j\qquad{\rm and}\qquad |v_i^*|^2+|v_j^*|^2 = |v_i|^2+|v_j|^2\ .
\end{equation}
Then, by the parallelogram law,
it follows that
\begin{equation}\label{crulesa}
|v_i^*-v_j^*|  = |v_i - v_j|\ .
\end{equation}
Given $v_i$ and $v_j$, the kinematically possible collisions of particles $i$ and $j$; i.e., those satisfying (\ref{kinpos}), may be parameterized in term of a unit vector $\sigma\in S^2$, the unit sphere in $\R^3$ as follows:
\begin{eqnarray}\label{crules}
v_i^*(\sigma) &=&  \frac{v_i+v_j}{2} +   \frac{|v_i -v_j|}{2}\sigma \nonumber\\
v_j^*(\sigma) &=&  \frac{v_i+v_j}{2} -   \frac{|v_i -v_j|}{2} \sigma  
\end{eqnarray}

The particular kinematically possible collision that occurs at time $T$ is selected according to the following rule:  There is given, in the specification of the process, a non-negative, even
function $b$ on $[-1,1]$ such that for any fixed $\sigma'\in S^2$, with $\dd \sigma$ denoting the uniform probability measure on $S^2$
\begin{equation}\label{normal}
\int_{S^2}b(\sigma\cdot \sigma')\dd \sigma = 1 \quad{\rm or, equivalently}, \qquad  \frac12 \int_{-1}^1 b(t){\rm d} t = 1 \ .
\end{equation}
 The example of main interest turns out to be
\begin{equation}\label{hardrate}
b(x) = 1\ .
\end{equation}
When $\alpha =1$  and $b$ is given by (\ref{hardrate}), 
the Kac process models ``hard sphere'' or ``billiard ball'' collisions \cite{K59}. (There are two standard parameterizations of the set of energy and momentum conserving collisions, the  ``$\sigma$ parameterization'' given by \eqref{crules}, and the ``$n$ parameterization''. While the latter is often used in physics texts and is used in \cite{K59}, the former, used here, has advantages. One is that in this parameterization, $b$ is constant, while in the other it is not due to a non-constant Jacobian relating the two parameterizations. See  Appendix A.1 of \cite{CCC} for more information; equation (A.18) of \cite{CCC} is the formula relating the $b$ functions for the two representations.)

In any case, 
as long as $v_i \neq v_j$, $b(\sigma\cdot(v_i-v_j)/|v_i-v_j|)$ is a probability density on $S^2$. At  time $T$,   $\sigma$ is selected  from the law 
$b(\sigma\cdot(v_i-v_j)/|v_i-v_j|)\dd \sigma$, and
then the process executes the collision step in which $v_i^*$ and $v_j^*$ are given by  (\ref{crules}). (If $v_i = v_j$, no jump is made.)
 Then, all of the waiting times are ``reset'' and the process begins afresh. 
This completes the probabilistic description of the one parameter family of  Kac collision process. 

This  one parameter family of Kac  collision process is a little more general than the one considered by Kac: 
There is an extra parameter $\alpha$ that ranges between $0$ and $2$. The case $\alpha = 0$ corresponds to Maxwellian 
molecules as in \cite{K56} or \cite{CGL}. The case $\alpha =1$ is the hard sphere case that is our main focus. The case $\alpha =2$ is the case 
of ``super hard spheres'' and estimates for this case will be useful in our study of $\alpha =1$.  Villani \cite{V03} discovered in the context of entropy production estimates  
that analysis of the non-physical case $\alpha =2$ could provide very helpful information on the physical cases $\alpha \leq 1$, and we make essential use of this insight in our analysis of spectral gaps.

\subsection{The generator of the Kac process}

\if false
An angle $\theta$ is selected from the distribution
$b(\cos\theta)\dd \theta$, and then $\vec v$ jumps to $R_{i,j,\theta}\vec v$
where
$$R_{i,j,\theta}\vec v = (v_1,v_2,\dots,v'_i(\theta),\dots,v'_j(\theta),\dots,v_n)$$
with 
$$v'_i(\theta) = v_i\cos(\theta) + v_j\sin(\theta)\qquad{\rm and}\qquad 
v'_j(\theta) = -v_i\sin(\theta) + v_j\cos(\theta)\ .$$
and $\theta \in (-\pi,\pi]$.
{\em Then, all clocks are reset to zero, and the process begins afresh.}   Let $\vec v(t)$ denote the state of the $N$-particle system at time $t$ as it evolves under this process.
\fi

The object of our investigation is the spectral gap for the generator of the Markov semigroup associated to this process. For {\em any} continuous  function $f$ on $\STE$, in particular without any symmetry assumption,
define 
$$\LN f (\vec v) = \frac1h \lim_{h\to 0}{\rm E}\{ f (\vec v(h))\ |\ \vec v(0) = \vec v\ \} \ .$$
We can write this more explicitly as
\begin{equation}\label{lndef2}
\LN f (\vec v) =  -{N}{\ncht}^{-1}\sum_{i<j} |v_i-v_j|^{\alpha}\left[f (\vec v) - [f ]^{(i,j)}(\vec v) \right] 
\end{equation}
where
\begin{equation}\label{lndef3}
 [f ]^{(i,j)}(\vec v)  =    \int_{S^2}b\left(\sigma\cdot \frac{v_i-v_j}{|v_i-v_j|}\right) f (R_{i,j,\sigma}\vec v) \dd \sigma
\end{equation}
and
${\displaystyle
(R_{i,j,\sigma}\vec v)_k = \begin{cases} v_i^*(\sigma) & k = i\\   v_j^*(\sigma) & k = j\\   v_k & k\neq i,j\end{cases}}$.
\


By (\ref{crulesa}) and (\ref{crules}), 
$$\cos\theta := \sigma\cdot \frac{v_i-v_j}{|v_i-v_j|} =   \frac{v_i^*-v_j^*}{|v_i^*-v_j^*|}\cdot \frac{v_i-v_j}{|v_i-v_j|}\ .$$
By this and  (\ref{crulesa}) once again, rates for the jump from $\vec v$ to $R_{i,j,\sigma}\vec v$ and from  $R_{i,j,\sigma}\vec v$ to $\vec v$
are equal. This is the property of ``detailed balance'' or ``microscopic reversibility''.  The analytic expression of this is self-adjointness of the generator $\LN$:

Let $\dd \sigma_N$ denote the uniform probability measure on $\STE$. (Note that $\STE$ is isometric to a sphere of radius $\sqrt{N(E - |p|^2)}$ in $\R^{3N -4}$, and by uniform, we mean uniform with respect to the symmetries of this sphere.)

For any two unit vectors $\sigma$ and $\omega$, one sees from (\ref{crules}) that 
\begin{equation}\label{Rdef4}
R_{i,j,\sigma}(R_{i,j,\omega}\vec v) =  R_{i,j,\sigma}\vec v\ .
\end{equation}
From this and the fact that the measure $\dd \sigma_N \otimes \dd \sigma$ is invariant under 
$$(\vec v,\sigma)  \mapsto (R_{i,j,\sigma}\vec v,(v_i-v_j)/|v_i - v_j|)\ ,$$
 it follows that for any two continuous functions $f$ and $g$ on $\STE$,
$$\langle g , \LN f \rangle_{L^2(\sigma_N)}  =
\langle  \LN g  ,\ f \rangle_{L^2(\sigma_N)}  \ ,$$
where $\langle \cdot,\cdot\rangle_{L^2(\sigma_N)}$ denotes the inner product on $L^2(\STE,\sigma_N)$. 
Thus,   $\LN$ is a self adjoint  operator on $L^2(\STE,\sigma_N)$.   Notice that the formulas (\ref{lndef2}) and (\ref{lndef3}) do not involve the parameters
$E$ and $p$, and hence  our notation references only $N$ and $\alpha$. 

Define the quadratic form $\E$ by
$\E(f ,f ) =  -\langle f  , \LN f \rangle_{L^2(\sigma_N)}$.
A simple computation using (\ref{Rdef4}) shows that
\begin{multline}\label{endefi}
\E(f ,f ) = \\
\frac{N}{2}{\ncht}^{-1}\sum_{i<j}\int_{\STE}\int_{S^2} |v_i-v_j|^{\alpha} b\left(\sigma\cdot \frac{v_i-v_j}{|v_i-v_j|}\right)
\left[  f (\vec v) - f (R_{i,j,\sigma}\vec v) \right]^2 
\dd \sigma \dd \sigma_N\ .
\end{multline}

One sees from this expression that $\LN$ is a negative semi-definite operator, and that provided $b$ is continuous at $1$,  $\LN f = 0$ 
if and only if $f $ is constant.  We are interested in the {\em spectral gap} of the operator $\LN$ on $L^2(\STE,\sigma_N)$:

\begin{equation}\label{Deldef1}
\Delta_{N,\alpha}(E,p)
= \inf\left\{ 
\E(f,f)
\ : \   \langle f,1 \rangle_{L^2(\sigma_N)} = 0\quad{\rm and}\quad \|f\|_{L^2(\sigma_N)}^2 =1\ \right\} \ .
\end{equation}

For fixed $N$, the dependence of $\Delta_{N,E,p}$ on $E$ and $p$ is quite simple:  Consider the point transformation
$$\phi_{E,p}(v_1,\dots,v_N) := \frac{1}{ \sqrt{E -|p|^2}}(v_1 - p, \dots, v_N - p)\ $$
that identifies $\STE$ with ${\mathcal S}_{N,1,0}$.  The induced transformation $U_{E,p}$ from $L^2({\mathcal S}_{N,1,0},\sigma_N)$ to 
$L^2(\STE,\sigma_N)$ given by
$U_{E,p}f = f\circ \phi_{E,p}$
is evidently unitary. A simple computation then shows that
\begin{equation}\label{scaling}
\E(U_{E,p}f,U_{E,p}f) = (E- |p|^2)^{\alpha/2}\E(f,f,)\ .
\end{equation}
As an immediate consequence, 
\begin{equation}\label{scaling2}
\Delta_{N,\alpha}(E,p)= (E- |p|^2)^{\alpha/2}  \Delta_{N,\alpha}(1,0)\ . 
\end{equation}

The dependence of $\Delta_{N,\alpha}(E,p)$ on $N$ is not so simple. Nonetheless, we have seen that the problem of  estimating  the quantity
$\Delta_{N,\alpha}(E,p)$ is essentially the same as the problem of estimating  $\Delta_{N,\alpha}(1,0)$. We therefore simplify our notation:

\begin{defi}[Spectral gap]\label{specgapdef} The {\em spectral gap for the $N$ particle Kac model} is the quantity
\begin{equation}\label{specgapdef1}
\Delta_{N,\alpha} :=   \Delta_{N,\alpha}(1,0)\ .
\end{equation}
\end{defi}

In what follows, we write $\ST$ to denote ${\mathcal S}_{N,1,0}$, and consider the Kac process on $\ST$ unless other values of $E$ and $p$
are explicitly specified. 
The {\em Kac conjecture for hard sphere collisions}  \cite{K59} is that 
$\liminf_{N\to \infty}\Delta_{N,1}> 0\ .$ Our main result shows somewhat more:

\begin{thm}[Spectral gap for the Kac Model with $0\leq \alpha \leq 2$]\label{main} For each continuous non-negative even  function $b$ on $[-1,1]$  statisfying \eqref{normal},
and for each $\alpha \in [0,2]$, there is a strictly positive constant $K$ depending only on $b$ and $\alpha$,  and explicitly computable, such that
$$\Delta_{N,\alpha}\geq K > 0$$
for all $N$.  In particular, this is true with $b$ given by \eqref{hardrate} and $\alpha =1$, the $3$ dimensional hard sphere Kac model.
\end{thm}

\subsection{The conjugate Kac process and its generator} 

Our method involves the introduction of another family of reversible Markov jump processes on $\ST$ that are {\em conjugate} to the Kac process.  For fixed $N$ and $\alpha$, this process is described as follows:  Given $\vec v\in \ST$, Let $\{ \widehat{T}_1, \dots,  \widehat{T}_N\}$ be $N$ independent exponential variables 
such that the parameter $\lambda_k(\vec v)$ of 
$\widehat{T}_k$ is
$$
\lambda_k(\vec v) = \frac{1}{N}\left[\frac{N^2 - (1+|v_k|^2)N}{(N-1)^2}\right]^{\alpha/2}\ .
$$
Since the total energy is $N$ and the total momentum is zero, the maximum possible value of $|v_k|^2$ on $\ST$ is $N-1$, and thus $\lambda_k \geq 0$, with 
equality only when $|v_k|^2$ takes on its maximal value.

The first jump time is $\widehat{T}= \min\{\widehat{T}_1,\dots,\widehat{T}_N\}$. At the jump time, if $k$ is the index furnishing the minimum, 
$\vec v$    jumps to a new point on $\ST$ such that $v_k$ is unchanged, but conditional on $v_k$, the other coordinates are redistributed uniformly.  
That is, the process makes a conditional jump to uniform, conditional on $v_k$ which is held fixed. After the jump, the process starts afresh. 
This completes the description of the {\em conjugate Kac process}. 

Note that the conjugate process is trivial for $N=2$, since then $v_2 = -v_1$, so that given one velocity, the other is known exactly, and the ``conditional jump to uniform'' is no jump at all in this case. However, alsready for $N=3$, the process is far from trivial.

\begin{remark}\label{conrem}  If $|v_k|^2$ is close to its expected value of $1$, then $\lambda_k(\vec v) \approx \frac{1}{N}$, 
which is exact for $\alpha = 0$. In this case, we have $N$ independent Poisson clocks with rate $\tfrac1N$ each, 
so that the mean waiting time for {\em some} jump is $1$. 

For $\alpha > 0$,  the rates $\lambda_k(\vec v)$ are not bounded away from $0$. However, 
{\em at most one of them can be very close to zero for any given state $\vec v$}. 
This is because, $\lambda_k(\vec v) = 0$ if and only if $|v_k|^2$  takes on its maximum value, $N-1$. 
For at most one value of $k$ is it possible that
$|v_k|^2 > \tfrac12 N$, and for $|v_j|^2 \leq \tfrac12 N$, $\lambda_j(\vec v) = \frac{1}{2N} + {\mathcal O}(\frac{1}{N^2})$.  Thus, for all $\alpha\in [0,2]$, for large $N$, 
the expected waiting time for a jump is very close to $1/N$, and this one jump will bring $N-1$ of the particles very close to equilibrium. If the expected waiting 
time were exactly $1/N$ and the jump took all $N$ particles to equilibrium, the spectral gap would be exactly $1- 1/N$. This is not misleading; we shall show 
that for the conjugate Kac process, the spectral gap is indeed $1-1/N$ plus lower order corrections. 
\end{remark}

To write down the generator, introduce the conditional expectation operators $P_k$, $k=1,\dots,N$, defined as follows:

For any function $\phi$ in $L^2(\ST)$, and any $k$ with $1\le k \le N$, define
$P_k(\phi)$ to be the orthogonal projection of $\phi$ onto the subspace of 
 $L^2(\ST)$ consisting of square integrable functions that depend on $\vec v$ through $v_k$ alone.   
 That is,  $P_k(\phi)$ is the unique element of  $L^2(\ST)$
of the form $f(v_k)$ such that
\begin{equation}\label{Pkfir0}
\int_{\ST}\phi(\vec v)g(v_k){\rm d}\sigma_N = 
\int_{\ST}f(v_k)g(v_k){\rm d}\sigma_N 
\end{equation}
for all continuous functions $g$ on $\R^3$.   In probabilistic language, $P_k\phi$ is the conditional expectation of $\phi$ given $v_k$:
\begin{equation}\label{Pkfir}
P_k \phi = E\{ \phi  \ :\ v_k \}\ .
\end{equation}

The generator of the conjugate Kac process is then given by
\begin{equation}\label{dKgen}
\widehat{L}_{N,\alpha}f  =  -  \frac{1}{N}\sum_{k=1}^N  \left[\frac{N^2 - (1+|v_k|^2)N}{(N-1)^2}\right]^{\alpha/2}[ f - P_k f]\ ,
\end{equation}
which is the analog of \eqref{lndef2}. 
Define the quadratic form ${\mathcal D}_{N,\alpha}$ by
$${\mathcal D}_{N,\alpha}(f ,f ) =  -\langle f  ,  \widehat{L}_{N,\alpha} f \rangle_{L^2(\sigma_N)}\ .$$
A simple computation using (\ref{Rdef4}) shows that
\begin{equation}\label{ddef2X}
\D(f,f) = \frac{1}{N}\sum_{k=1}^N  \int_{\ST}\left[\frac{N^2 - (1+|v_k|^2)N}{(N-1)^2}\right]^{\alpha/2} 
\left[f^2 - f P_k f\right]  \dd \sigma_N\ .
 \end{equation}
 
 The spectral gap for the conjugate Kac process is the quantity defined by
\begin{equation}\label{Deldef1}
\widehat{\Delta}_{N,\alpha}
= \inf\left\{ 
{\mathcal D}_{N,\alpha}(f,f)
\ : \ ,\  \langle f,1 \rangle_{L^2(\ST)} = 0\quad{\rm and}\quad \|f\|_{L^2(\ST)}^2 =1\ \right\} \ .
\end{equation}

The following theorem bears out the heuristic discussion in Remark~\ref{conrem}

 \begin{thm}\label{thm2X}  For all $N\geq 3$, and all $\alpha\in [0,2]$, $\widehat{\Delta}_{N,\alpha} > 0$. 
 Moreover, there is a constant $C$ independent of $N$ such that 
 \begin{equation}\label{ind2X}
 \widehat{\Delta}_{N,\alpha} \geq 1 - \frac{1}{N} - \frac{C}{N^{3/2}}\ .
\end{equation}
\end{thm}

\begin{remark} The constant $C$ is large enough that the first statement does not follow from \eqref{ind2X} which is only a meaningful bound when $N$ is large enough that the right side is positive. 
\end{remark}

\subsection{The link between the Kac process and its conjugate} 
The following theorem provides the link between the Kac process and its conjugate:

 \begin{thm}\label{thm1X}  For all $N\geq 3$,
 \begin{equation}\label{ind1X}
 \Delta_{N,\alpha}  \geq \frac{N}{N-1}\Delta_{N-1,\alpha}\widehat{\Delta}_{N,\alpha}\ .
\end{equation}
\end{thm}

Before proving Theorem~\ref{thm1X}, we recall some explicit formulas that will be useful here and elsewhere. The proof of Theorem~\ref{thm1X} uses the methods introduced in \cite{CCL00,CCL03,CCL14}.
The estimation of $\DN$ in terms of $\Delta_{N-1,\alpha}$ is based  on a parameterization of $\ST$, for $N\geq 3$,   in terms of ${\mathcal S}_{N-1}\times B$
where $B$ is the unit ball. For each $k=1,\dots,N$, define $\pi_k: \ST \to B$ by
\begin{equation}\label{kproj}
\pi_k(\vec v) = \frac{1}{\sqrt{N-1}}v_k\ .
\end{equation}
(Note that because of the constraints $\sum_{j=1}^N v_j = 0$ and  $\sum_{j=1}^N |v_j|^2= N$, the largest value of $|v_k|$ on $\ST$ is $\sqrt{N-1}$.)

Define a map $T_1: {\mathcal S}_{N-1}\times B \to \ST$ as follows:
\begin{equation}\label{factor}
T_1(\vec y,v) = 
\left(  \sqrt{N-1}v\ , \ \beta(v) y_1 - {1\over \sqrt{N-1}}v, \dots,\beta(v) y_{N-1} -
{1\over \sqrt{N-1}}v\right)\ ,
\end{equation}
where
\begin{equation}\label{factor2}
\beta^2(v) = \frac{N }{N-1}(1- |v|^2)\ .
\end{equation}
The subscript $1$ in $T_1$ indicates that the vector $v$ from $B$ went into the first place. We likewise define $T_2,\dots,T_N$ by placing this coordinate in the corresponding position.

In the  coordinates $(\vec y,v)$ on $\ST$ induced by any of the maps $T_k$,  one has
the integral factorization  formula 
\begin{equation}\label{facform}
\int_{\ST}\phi(\vec v){\rm d}\sigma_N = 
\int_{B}\left[\int_{{\mathcal S}_{N-1}}\phi(T_k(\vec y,v)) {\rm d}\sigma_{N-1}\right]
\dd \nu_N(v)\ .
\end{equation}
where for all $N\geq 3$, 
\begin{equation}\label{nun}
\dd \nu_N(v) = 
{|S^{3N-7}|\over|S^{3N-4}|}(1 -|v|^2)^{(3N-8)/2} {\rm d}v\ .
\end{equation}

Also, note that 
for $i\ne k,j\ne k$, 
\begin{equation}\label{inter}
R_{i,j,\sigma}(T_k(\vec y,v)) = T_k(R_{i,j,\sigma}(\vec y),v)\ .
\end{equation}
We now have the means to relate $\E$ to $\Em$.

For each $k=1,\dots,N$, define the {\em conditional Dirichlet form} $\E(f,f|v_k)$  on $L^2(\ST,\sigma_N)$
by
\begin{multline}\label{cdir}
\E(f,f|v_k) =
(N-1){\nmcht}^{-1}\times \\\sum_{i<j; i,j\neq k}\int_{{\mathcal S}_{N-1}}\int_{S^2} 
|y_i-y_j|^{\alpha} b\left(\sigma\cdot \frac{y_i-y_j}{|y_i-y_j|}\right)
F^2(\vec v,y) 
\dd \sigma \dd \sigma_{N-1}(y)\ .
\end{multline}
where
$F(\vec v,y) :=\left[  f (T_k(\pi_k(\vec v),y)  - f (R_{i,j,\sigma}T_k(\pi_k(\vec v),y)) \right]$.

As the integration on the right is only over the ``slices'' of $\ST$ at constant values of $v_k$,
the result is still a non-trivial function of $v_k$. 
For each fixed $v_k$, the conditional Dirichlet form is simply the $N-1$ particle Dirichlet form acting  in the $\vec y$ variables.

Note that by (\ref{factor}) and (\ref{factor2}), when $\vec v = T_k(\vec y,v) $ and $i,j\neq k$,
$$|v_i - v_j |^2= \beta^2(\pi_k(\vec v))|y_i - y_j|^2  =  \frac{N^2 - (1+|v_k|^2)N}{(N-1)^2} |y_i - y_j|^2 \ .$$
We define, for $v\in \R^3$, $|v|^2 \leq N-1$, 
\begin{defi}
\begin{equation}\label{weightdef}
w_N(v) := \frac{N^2 - (1+|v|^2)N}{(N-1)^2}\ .
\end{equation}
\end{defi}
We therefore have that 
$$|v_i-v_j|^{\alpha} b\left(\sigma\cdot \frac{v_i-v_j}{|v_i-v_j|}\right) =  w^{\alpha/2}_N(v_k)
|y_i-y_j|^{\alpha} b\left(\sigma\cdot \frac{y_i-y_j}{|y_i-y_j|}\right)\ .$$

 Then, using (\ref{inter}), one easily checks that
\begin{equation}\label{rec}
\E(f,f) = \frac{N}{N-1}
\left(\frac{1}{N}\sum_{k=1}^N  \int_{B}w^{\alpha/2}_N(v_k) \E(f,f|v_k) \dd \nu_N(v_k/\sqrt{N-1})\right)\ .
\end{equation}

\subsection{Proof of Theorem~\ref{thm1X}}

\begin{proof}[Proof of Theorem~\ref{thm1X}]
To estimate the right hand side  of \eqref{rec}  in terms of $\Delta_{N-1,\alpha}$, we must take into account that for fixed $v_k$, $f$ need not  be orthogonal to the constants
as a function of the remaining variables $\vec y$. To take this  into account, we use the projection operators already introduced in \eqref{Pkfir0} and \eqref{Pkfir}. Using  
 the factorization formula (\ref{facform}), we have an explicit formula:
$$P_k\phi(\vec  v) = \int_{{\mathcal S}_{N-1}}\phi(T_k(\vec y,v_k/\sqrt{N-1})) {\rm d}\sigma_{N-1}\ ,$$

Now note that 
$\E(f,f|v_k)  = \E(f- P_kf,f-P_kf |v_k)$,
and then using the spectral gap for $N-1$ particles and (\ref{rec}), 
 one has
\begin{eqnarray}\label{rec2H}
\E(f,f) &\geq& \frac{N}{N-1}\Delta_{N-1, \alpha}
\left(\frac{1}{N}\sum_{k=1}^N  \int_{\ST}w^{\alpha/2}_N(v_k)
\left[f - P_k f\right]^2  \dd \sigma_N\right)\nonumber\\
&=&  \frac{N}{N-1}\Delta_{N-1, \alpha} \D(f,f)
\end{eqnarray}
since
\begin{multline*}\frac{1}{N}\sum_{k=1}^N  \int_{\ST}w^{\alpha/2}_N(v_k)
\left[f - P_k f\right]^2  \dd \sigma_N  = \\
\frac{1}{N}\sum_{k=1}^N  \int_{\ST}w^{\alpha/2}_N(v_k)
\left[f^2 - f P_k f\right]  \dd \sigma_N = \D(f,f)\ .
\end{multline*}
The theorem follows directly from \eqref{rec2H} and the variational characterizations of 
$\Delta_{N,\alpha}$ and $\widehat{\Delta}_{N,\alpha}$
\end{proof}

\subsection{Proof of the main theorem}

Combining Theorem~\ref{thm1X} and Theorem~\ref{thm2X} yields, for a  constant $C$, independent of $N>3$, 
\begin{equation}\label{indulink}
\Delta_{N,\alpha}  \geq \left(1 - \frac{C}{N^{3/2}}\right) \Delta_{N-1,\alpha}
\end{equation}
The main result will follow easily from this, and a bound on $\Delta_{2,\alpha}$, and our next task is to prove such a bound. 

 Because
$|v_i-v_j|$ can be arbitrarily small on $\ST$ for any $N>2$, for any given $C>0$,  there will be functions
 $f\in L^2(\sigma_N)$ that satisfy $\langle f,1\rangle_{L^2(\sigma_N)} = 0$ and $\|f\|^2_{L^2(\sigma_N)} =1$
such that
$ f(\vec v) \LN f(\vec v) < C    f(\vec v) L_{N,0} f(\vec v)$
for some $\vec v\in \ST$.   This precludes a simple and direct comparison of the Dirichlet forms $\E$ and ${\mathcal E}_{N,0}$.

\if false
When $\alpha>0$,  the parameter $\lambda_{i,j}$ is small when $|v_i- v_j|$ is small. Pairs of particles with a small relative velocity will collide only infrequently. 
Since it is the collision mechanism that drives the process to equilibrium, the suppression of collisions for such pairs is a problem to be overcome. Of course,
if  $|v_i- v_j|$ is small, there will be some $k$ so that   $|v_i- v_k|$ and  $|v_k- v_j|$ are not small, and eventually  some such particle  will collide with particle $j$, say, 
and then after the collision, the difference in velocities of  particles $i$ and $j$ is unlikely to be small. Thus, it is natural to expect that the smallness of 
$|v_i- v_j|$ for certain pairs of particles is not a significant impediment to the existence of a spectral gap, but it does prevent us from making a simple comparison with the
$\alpha =0$ case for $N > 2$
\fi

For $N=2$, things are much better: Then by definition of ${\mathcal S}_2 := {\mathcal S}_{2,1,0}$, for all $(v_1,v_2)\in  {\mathcal S}_2$, $v_2 = -v_1$,
and $|v_1| = |v_2| =1$, so that $|v_1-v_2| = 2$ everywhere on ${\mathcal S}_2$. That is, for $N=2$, there is no significant difference between $\alpha = 0$ and $\alpha> 0$. 
For $\alpha =0$ and a number of choices of $b$, $\Delta_2$ has been computed in \cite{CGL}.  The following is proved in Lemma~2.1 of \cite{CGL}

\begin{lm}[Spectral gap for $N=2$ and hard sphere collisions ]\label{sg2}With $b(x) = 1$, 
\begin{equation}\label{2gaphs}
\Delta_{2,1} = 2\ .
\end{equation}
\end{lm}

The proof given in \cite{CGL} is fairly simple, and  it is easy to apply the formulas there to other choices for the probability density $b$, and to  show that  as long as $b$ is even and continuous on $[-1,1]$, $\Delta_{2,\alpha} > 0$  for all $\alpha\in [0,2]$.

 We are now ready to prove the main theorem:
 
 \begin{proof}[Proof of Theorem~\ref{main}]   Since $\Delta_{2,\alpha} > 0$ by Lemma~\ref{sg2},  Theorem~\ref{thm1X} and the first part of
  Theorem~\ref{thm2X} yield
  $$\Delta_{3,\alpha}\geq \frac32 \Delta_{2,\alpha} \widehat{\Delta}_{3,\alpha} > 0\ ,$$
 and then the obvious iteration yields $\Delta_{N,\alpha}> 0$ for all $N\geq 2$. To go further and prove that $\inf_{N\geq 2}\Delta_{N,\alpha} > 0$, we use the 
 second part of   Theorem~\ref{thm2X}:

Let $N_0$ be such that ${\displaystyle 1 - C N_0^{-3/2} > 0}$. Then
${\displaystyle 
K_0 := \prod_{j= N_0}^\infty \left(1 - \frac{C}{j^{3/2}}\right)  > 0}$
and for all $N\geq N_0$,
$\Delta_{N,\alpha}  \geq K_0 \Delta_{N_0,\alpha}$.
\end{proof}

\begin{remark}  As we shall see, it is possible to explicitly compute the constant $C$ in Theorem~\ref{thm2X}. To keep the presentation free of clutter, we have not carried this through here, but it would be a simple, if tedious, exercise to track the constants step by step. As for the first part of Theorem~\ref{thm2X}, it is easy to give an explicit lower bound on $\widehat{\Delta}_{N,\alpha}$ for all $N \geq 4$, and we do so below. The case $N=3$ is more difficult, and we use a simple compactness argument to prove $\widehat{\Delta}_{3,\alpha}>0$. However we do sketch a method for explicitly estimating $\widehat{\Delta}_{3,\alpha}$. Thus, the method we employ to prove Theorem~\ref{main} can be used to prove explicit bounds. 
\end{remark}

It remains to prove Theorem~\ref{thm2X}, and we prepare the way for this in the next section. Throughout the rest of the paper, we are concerned soley with the conjugate Kac process. All of the analysis that directly involves the Kac process itself is complete at this point.

\section{Estimates for the conjugate process}

It is in the proof of Theorem~\ref{thm2X} that new ideas are required to deal with the non-uniform jump rates of the conjugate process, and we begin with a heuristic discussion of these ideas. 

As in the case $\alpha =0$, we rely in part on 
the fact that 
the invariant measure $\sigma_N$ 
(of both processes) is {\em chaotic} in the sense of Kac. 
More specifically, it is 
$\gamma$ chaotic where 
$${\rm d}\gamma = (2\pi /3)^{-3/2} e^{-3|v|^2/2}{\rm d}v$$
is the isotropic Gaussian distribution on $\R^3$ with unit variance. This means that for any $k\in \N$ and any bounded continuous function 
$\psi(v_1,\dots,v_k)$ on  $\R^{3k}$, 
$$\lim_{N\to\infty} \int_{\ST} \psi(v_1,\dots,v_k){\rm d}\sigma_N  = \int_{\R^{3k}}\psi(v_1,\dots,v_k){\rm d}\gamma^{\otimes k}\ .$$
That is, as long as $k$ is much less than $N$, the random variables $v_1,\dots,v_k$ are nearly independent, and by symmetry this is true of any set of $k$ 
distinct coordinate functions on $\ST$.  The notion of chaos was also introduced by Kac in \cite{K56}, and the main result of that paper was that for 
the model with one dimensional velocities and $\alpha =0$, chaos is propagated by the dynamics. Propagation of chaos for $\alpha > 0$ is much harder, 
and this was only proved later by Sznitman \cite{Sznit}, also in $3$ dimensions. 

In case $\alpha =0$, the range of $I -\widehat{L}_{N,0}$  has a special structure that facilitates the study of the spectral gap for $\widehat{L}_{N,0}$.  The subspace of the range that is orthogonal to the constants  consists of functions $f$ of the form:
$f(\vec v) = \sum_{j=1}^N \varphi_j(v_j)$ such that each $\varphi_j(v_j)$ is square integrable and such that $f$ is orthognal to the constants. One choice for the $\varphi_j$'s is $\varphi_j = P_jf$, but there are other choices: Since  for any fixed $a\in \R^3$ and $b\in \R$, 
\begin{equation}\label{conlawcon}
\sum_{j=1}^N(a\cdot v_j + b(|v_j|^2-1)) = 0\, ,
\end{equation}
we may make the replacement $\varphi_j(v_j) \longrightarrow  \varphi_j(v_j)+ a\cdot v_j + b(|v_j|^2-1)$ without changing $f(\vec v)$. There is however, a prefered choice of the functions $\varphi_j$ that plays an important role in what follows. As we shall show, there is a unique choice that minimizes $\sum_{j=1}^N \|\varphi_j(v_j)\|_2^2$ which has a number of useful properties.

\begin{defi}\label{ANdef}  Let ${\mathcal A}_N$ denote the subspace of $L^2(\ST)$ that is the closure of the span of functions of the form
\begin{equation}\label{ANdef1}
f(\vec v) = \sum_{j=1}^N \varphi_j(v_j)
\end{equation}
for bounded continuous functions $\varphi_1, \dots,\varphi_N$ in $\R^3$ such that $\int_{\ST}f{\rm d}\sigma_N =0$.   
\end{defi}

When $\alpha \neq 0$, ${\mathcal A}_N$ is not an invariant subspace  of $\widehat{L}_{N,\alpha}$.  Nonetheless, as we explain, the gap may be bounded using a trial function decomposition 
based on  ${\mathcal A}_N$, and for this the approximate independence that comes along with the chaoticity of $\sigma_N$ is essential. 

To see how this works,
suppose that one replaces the state space $\ST$ with $\R^{3N}$, and replaces the conjugate Kac process  with the ``conditional jump to uniform" process with respect to ${\rm d}\gamma^{\otimes N}$.  
In this case, with the invariant measure being a product measure, the corresponding conditional expectation operators $P_k$ will all commute.  
One might therefore expect that the operators $P_k$ figuring in the definition \eqref{dKgen} of $\widehat{L}_{N,\alpha}$ almost commute for large $N$. 
Suppose that they {\em exactly} commute, or, what is the same thing, that the coordinate functions $v_1,\dots,v_N$ are {\em exactly} independent.

Since $0 = \int_{\ST}f{\rm d}\sigma_N  = \sum_{j=1}^N \int_{\ST}\varphi_j(v_j){\rm d}\sigma_N =0$, replacing $\varphi_j(v_j)$ by 
$\varphi_j(v_j) - \int_{\ST}\varphi_j(v_j){\rm d}\sigma_N$, we may assume without loss of generality in \eqref{ANdef1} that $\int_{\ST}\varphi_j(v_j){\rm d}\sigma_N = 0$ for 
each $j$.  Granted the exact independence, we would then have that for $k\neq j$, $P_k \varphi_j(v_k) = 0$, while $P_k\varphi_k(v_k) = \varphi_k(v_k)$. 
Thus, for $f\in {\mathcal A}_N$, $f - P_kf = \sum_{j\neq k} \varphi_j(v_j)$, and then, again using the independence,
 \begin{eqnarray}\label{hopefulA}
 \D(f,f) &=& \frac{1}{N}\sum_{k=1}^N \int_{\ST}   w_N^{\alpha/2}(v_k)   \sum_{j\neq k} \varphi_j^2(v_j) \dd \sigma_N\nonumber\\
 &=& \frac{1}{N}\sum_{k=1}^N \left(\int_{\ST}   w_N^{\alpha/2}(v_k)\dd \sigma_N  \right) 
 \left(\int_{\ST}  \sum_{j\neq k} \varphi_j^2(v_j) \dd \sigma_N\right)
 \end{eqnarray}
 As is shown below, the integral over the rate, which is evidently independent of $k$, is bounded below by $1 - C/N^2$ for some constant $C$ that is independent of $N$.
 Thus, we would have
\begin{eqnarray}\label{rx1}
 \D(f,f) &\geq& \left(1 - \frac{C}{N^2}\right) \frac{N-1}{N} \sum_{j=1}^N\|\varphi_j(v_j)\|_2^2\nonumber\\
  &=& \left(1 - \frac{C}{N^2}\right) \frac{N-1}{N}\|f\|_2^2\ 
 \end{eqnarray}
 which is even better than \eqref{ind2X}. 
 
The equality in \eqref{rx1} comes from the identity $\sum_{j=1}^N\|\varphi_j(v_j)\|_2^2 = \|f\|_2^2$ which is true when there is {\em exact} independence
 of the coordinate functions.  In our setting, we do not have exact independence, and must prove and use appropriate {\em quantitative chaos} estimates.  For instance, in \eqref{minchar} of Theorem~\ref{gprop2}, it is shown that in our  setting, for a particular decomposition $f(\vec v) = \sum_{j=1}^N \varphi(v_j)$ -- such decompositions are not unique, even if one requires each $\varphi_j$ to be orthogonal to the constants -- one has 
\begin{equation}\label{rx2}
 \sum_{j=1}^N\|\varphi_j(v_j)\|_2^2 \leq \left(1+\frac{C}{N^2}\right) \|f\|_2^2
 \end{equation}
 for all $N\geq 3$ and with $C$ independent of $N$.  Using this after the first inequality in \eqref{rx1} still yeilds someting even better than \eqref{ind2X}. .

 Of course, one must consider trial functions that are not in ${\mathcal A}_N$,
and for trial functions $f$  that are  
 in ${\mathcal A}_N^\perp$, things are better still.  Such functions are shown to belong to the null space of $P_k$ for each $k$.  Therefore, for $f\in {\mathcal A}_N^\perp$, we would have from \eqref{ddef2X}
 \begin{equation}\label{hopeful}
 \D(f,f) = \int_{\ST} \left(\frac{1}{N}\sum_{k=1}^N  w_N^{\alpha/2}(v_k) \right)
f^2 \dd \sigma_N\ ,
 \end{equation}
 which is a significant simplification of \eqref{ddef2X}.  It is shown below (see Lemma~\ref{weight} and Remark~\ref{lowN}) that for some constant $C$ independent of $N$,
 $$
 \frac{1}{N}\sum_{k=1}^N  w_N^{\alpha/2}(v_k)  \geq 1- \left(1 -\frac{\alpha}{2} \right)\frac{1}{N} - \frac{C}{N^2}\ .
 $$
 Combining this with \eqref{hopeful} would then yield
 $$
  \D(f,f)  \geq \left(1-\left(1 -\frac{\alpha}{2} \right)\frac{1}{N} - \frac{C}{N^2}\right)\|f\|_2^2\ .
 $$
 For $\alpha>0$, this is much stronger  than \eqref{ind2X}, and for this bound we do not even use the approximate independence.  
 
 Since ${\mathcal A}_N$ is not an invariant subspace for $\widehat{L}_{N,\alpha}$,  one has to show that  for $g\in {\mathcal A}_N$ and $h\in  {\mathcal A}_N^\perp$,
 $D(g,h)$ is small. We shall show, again using the approximate independence, that 
 $$|D(g,h)| \leq \frac{C}{N^{3/2}}\|g\|_2\|h\|_2\ .$$
 It is the estimate in this step that is responsible for the $N^{3/2}$ term in \eqref{indulink}. A more refined argument, like the one provided for this step in \cite{CCL14} for the model with one dimensional velocities, would presumably improve $N^{3/2}$ to $N^2$, but since we have elected not to keep track of constants, there is no point in pursuing this here. 
 
Our proof will closely follow these heuristics, but of course we must carefully control the departures from exact independence wherever it was used above. There is one significant twist. Though we have the estimate \eqref{rx2}, it is much easier to prove the weaker analog of it with $N^2$ replaced by $N$, and perhaps there are other models that are ``less chaotic'' in which the weaker bound is all that one has. The weaker bound cannot be used directly in \eqref{rx1} to obtain anything useful, but the simple device of defining  ${\mathcal F}_{N,\alpha} := \|f\|_2^2 - \D(f,f)$ reduces the problem of estimating the gap for $\D$  to that of obtaining an appropriate  of an upper bound for ${\mathcal F}_{N,\alpha}$, and with the $1$ out of the way, the weaker version of \eqref{rx2} becomes useful.  This is what we do in Section~\ref{sec3.3} to complete the proof. 

The next subsection presents the ``quantitative chaos'' estimates that are used to control  the weak dependence of the coordinate function for large $N$.  It is important that some of the results turn out to be meaningful even for small $N$, such as $N=3$.

\subsection{Quantitative chaos}

A number of the quantitative chaos bounds that we need may be expressed in terms of the {\em correlation operator} $K$-operator that we now define:

Let $B$ denote the unit ball in $\R^3$. Let $N\geq 3$ and let $\nu_N$ be given by \eqref{nun}, so that
for any function $\psi$ on $B$, and any $k$, 
$$\int_B \psi(v){\rm d}\nu_N = \int_{\ST} \psi(\pi_k(\vec v)){\rm d}\sigma_N\ ,$$
where $\pi_k$ is given by \eqref{kproj}, and $\nu_N$ is given by \eqref{nun}.
We define the  operator $K$ on $L^2(B,\nu_N)$ by 
\begin{equation}\label{Kopdef}
\langle \psi_1,K \psi_2\rangle_{L^2(B,\nu_N)} = \int_{\ST} \psi_1^* (\pi_1(\vec v))\psi_2(\pi_2(\vec v)){\rm d}\sigma_N\ .
\end{equation}
$K$ is evidently self adjoint. 

\begin{defi}\label{xidef}For $j=0,\dots,4$ define  fnctions $\xi_j(v)$ on $B$ by
\begin{equation}\label{etaeigs}
\xi_0(v) =1\quad \xi_j(v) = v_j\ , j =1,2,3, \quad{\rm and}\quad \xi_4(v) = (|v|^2-1)/(N-1)\ .
\end{equation}
\end{defi}

The spectrum of $K$ is determined in \cite{CGL}, where the following facts are proven:

\begin{lm}\label{CGL} Let $N\geq 3$.  The operator $K$ is compact.   The function $\xi_0$ is an eigenfunction of $K$
 with eigenvalue $1$, and it spans the corresponding eigenspace. The functions  $\xi_j$, $j=1,2,3,4$ are eigenfunctions of $K$ with eigenvalue  $-1/(N-1)$,
 and they are an orthogonal basis for this eigenspace. 
 No other eigenvalues of $K$ are larger in absolute value than $\tfrac{5N-3}{3(N-1)^3}$.  Therefore, for all $\psi_1,\psi_2\in L^2(B,\nu_N)$ that are orthogonal to the constants, 
 the three components of $v$ and $v^2$, 
 \begin{equation}\label{kbnd}
\left|  \int_{\ST} \psi_1^*(\pi_1(\vec v))\psi_2(\pi_2(\vec v)){\rm d}\sigma_N\right|  \leq   \frac{5N-3}{3(N-1)^3} \|\psi_2\circ\pi_1\|_2 \|\psi_2\circ \pi_2\|_2\ .
 \end{equation}
 Equivalently,   for all functions $\psi\in L^2(B,\nu_N)$, that are orthogonal to $1$, the three components of $v$ and $v^2$, 
 \begin{equation}\label{ktop} 
\|K\psi\|_2 \leq \frac{5N-3}{3(N-1)^3} \|\psi\|_2\ .
 \end{equation}
Finally, every eigenvalues $\kappa$ of $K$, other than $1$,  $\frac{5N-3}{3(N-1)^3}$ and $\frac{1}{N-1}$ staisfies
 \begin{equation}\label{upperlow}
 -\frac{7N-3}{3(N-1)^4} \leq \kappa  <   \frac{5N-3}{3(N-1)^3}\ .
 \end{equation}
 \end{lm}
 
 \begin{remark} The number on the left in \eqref{upperlow} is the eigenvalue denoted by $\kappa_{1,2}$ in Section 8 of \cite{CGL}.
 \end{remark}

 Fix some $k$, and let $\mathcal{H}$  the subspace  of $L^2(\ST)$  spanned by functions of the form $\varphi(v_k)$ for some $k$. Since $v_k$ ranges over the ball of radius $\sqrt{N-1}$ in $\R^3$, one may think of 
 $\mathcal{H}$ as a Hilbert space consisting of square integrable functions on this ball, with respect to a scaled version of the measure $\nu_N$. 
 
 It will be convenient in what follows to think of $K$ as an operator on $\mathcal{H}$. Note that  $\varphi(v_k) = \tilde \varphi(\pi_k(\vec v))$ where $\tilde \varphi(v) = \varphi(\sqrt{N-1} v)$.
 Define 
 \begin{equation}\label{spheretoball}
 K\varphi(v_k) = (K\tilde \varphi)(\pi_k(\vec v)) \ .
 \end{equation}  
 The spectrum  of $K$, including multiplictiy,  thought of this way is naturally the same, but the eigenfunctions change by scaling. For example, now $|v|^2 -1$ is an eigenfunction with eigenvalue $-1/(N-1)$.   In this notation, we have that for any function $\xi$ on $\R^3$ so that $\xi(v_1)$ is in $L^2(\sigma_N)$,
\begin{equation}\label{ScaledK}
E\{ \xi(v_1) \ |\  v_2 = v\} = K\xi(v_N)\ .
\end{equation}
The $K$ operator defined by \eqref{ScaledK} is simply a ``scaled'' version of the $K$ operator defined in \eqref{Kopdef}, scaled so it operates on functions on $\mathcal{H}$. For some computations, particularly in the computation of eigenvalues of $K$, the definition \eqref{Kopdef} is more convenient. For  other computations, more directly connected the the Kac process, \eqref{ScaledK} has advantages.
This slight abuse of notation will simplify many formulas that follow without introducing any ambiguity. 

Since $K$ is compact,  there is a  orthonormal basis of 
$\mathcal{H}$ consisting of eigenvectors of $K$. This  orthonormal basis is determined explicitly in \cite{CGL}, but all we need to know
is that is can be written as $\{\eta_\iota\}_{\iota\geq 0}$ where 
\begin{equation}\label{etabasis}\eta_0(v) =1\ , \  \eta_j(v) = \sqrt{3}{\bf e_j}\cdot v \ , \ 1 \leq j\leq 3 \ {\rm and} \  \eta_4(v) = C_N( |v|^2 -1)\ ,
\end{equation}
 with $C_N$ being a normalization
  constant. This follows directly from Lemma~\ref{CGL} and \eqref{spheretoball}.

Let $\kappa_\iota$ denote the eigenvalue corresponding to $\eta_\iota$, so that $K\eta_\iota = \kappa_\iota\eta_\iota$. 
Our first application of Lemma~\ref{CGL} concerns the norm of functions in ${\mathcal A}_N$:  

\begin{thm}\label{gprop2}
 Let $N\geq 3$, and let $f \in {\mathcal A}_N$ be  orthogonal to $1$. Then there is a unique choice of $\varphi_1,\dots,\varphi_N$ with $f =  \sum_{j=1}^N \varphi_j(v_j) $  and each $\varphi_j(v_j)$ orthogonal to the constants
 that minimizes $\sum_{k=1}^N\|\varphi_k\|_2^2$  where $\|\varphi_k\|_2^2$ denotes $\int_{\ST}|\varphi_k(v_k)|^2{\rm d}\sigma_N$. Let 
 \begin{equation}\label{goodexp}
 \varphi_j(v_j)= \sum_{i=1}^\infty a_{j,i}\eta_i(v_j)
 \end{equation}
 be the expansion of $\varphi_j$ in the orthonormal basis consisting of eigenfunctions of $K$ that is specified above. Then this minimizer is characterized by 
 \begin{equation}\label{minchar}
 \sum_{j=1}^Na_{j,i} = 0 \qquad{\rm for}\qquad  1 \leq i \leq 4\ .
 \end{equation}
For this choice,
 \begin{equation}\label{gp3}
\left(1 - \frac{7N-3}{3(N-1)^3}\right)\sum_{k=1}^N\|\varphi_k\|_2^2  \leq \|f\|_2^2 \leq 
 \left(1 +\frac{5N-3}{3(N-1)^2}\right)\sum_{k=1}^N\|\varphi_k\|_2^2\ ,
 \end{equation}
 In particular, let ${\mathcal H}_{N,k}$ denote the subspace of
 $L^2(\ST)$ consisting of functions of the form $\varphi(v_k)$.  Define  ${\mathcal B}_N$ to be the subspace of $\bigoplus_{k=1}^N {\mathcal H}_{N,k}$ consisiting of $(\varphi_1,\dots,\varphi_N)$ such that \eqref{minchar} is satisfied.  Then
 the operator $T: {\mathcal B}_N \to {\mathcal A}_N$ defined by 
 by
 $$T(\varphi_1(v_1),\dots,\varphi_N(v_N)) = \sum_{k=1}^N \varphi_k(v_k)$$
 is bounded with a bounded inverse. 
 \end{thm} 

\begin{remark} Define $c_N := \left(1 - \frac{7N-3}{3(N-1)^3}\right)$ and $C_N = \left(1 +\frac{5N-3}{3(N-1)^2}\right)$. Note the different exponents in the denominator, and that $c_N > 0$ for all $N\geq 3$. Also note that $c_N = 1 - {\mathcal O}(1/N)$, and 
$C_N = 1 + {\mathcal O}(1/N)$, and then we can rewrite \eqref{gp3} as
$$c_N \sum_{k=1}^N\|\varphi_k\|_2^2  \leq \|f\|_2^2 \leq C_N   \sum_{k=1}^N\|\varphi_k\|_2^2 \ ,$$
and of course we would have equality here with $c_N = C_N =1$ if the coordinate functions were exactly independent.  Theorem~\ref{gprop2} gives a quantitative expression of the fact that 
 for large $N$, the coordinate functions are approximately pairwise independent.  
\end{remark}

\begin{proof}[Proof of Theorem~\ref{gprop2}]
As noted above, we may assume that each $\varphi_j$ is orthogonal to the constants.   We expand each $\varphi_j$ in the eigenbasis of $K$ as follows:
\begin{equation}\label{ioex1}
\varphi_j(v_j) = \sum_{\iota=1}^\infty a_{j,\iota} \eta_\iota(v_j)\ .
\end{equation}
Then evidently ${\displaystyle \|\varphi_j\|_2^2 = \sum_{\iota=1}^\infty |a_{j,\iota}|^2}$.  On account of \eqref{conlawcon}, for $j=1,2,3,4$, we may replace $a_{j,i}$ by $a_{j,i} -t_i$ without changing $f(\vec v) = \sum_{j=1}^N\varphi_j(v_j)$. With this modification, 
$$\sum_{j=1}^N \|\varphi_j\|_2^2 = \sum_{j=1}^N\sum_{i=1}^4 (a_{j,i} -t_i)^2  + \sum_{j=1}^N\sum_{i=5}^\infty|a_{i,j}|^2\ ,$$
which is evidently minimized by taking $t_i = - \frac1N\sum_{j=1}^Na_{j,i}$.  and then making this replacement, \eqref{minchar} is satisfied.

Next,  for $j\neq k$,
$$\int_{\ST} \varphi^*_j(v_j) \varphi_k(v_k){\rm d}\sigma_N = \sum_{\iota,\iota' =1}^\infty a^*_{j,\iota}a_{k,\iota'}\langle \eta_{j,\iota}, K\eta_{k,\iota'}\rangle_{\mathcal{H}}  = \sum_{\iota=1}^\infty \kappa_\iota a^*_{j,\iota}a_{k,\iota}\ .
$$
Therefore, when $f$ is given by  \eqref{ANdef1} and \eqref{ioex1} with \eqref{minchar} satisfied,
\begin{equation}\label{ioexA}
\|f\|_2^2 =  \sum_{\iota=1}^\infty \left(\sum_{j=1}^N |a_{j,\iota}|^2 + \kappa_\iota \sum_{j\neq k, j,k=1}^N  \Re a^*_{j,\iota}a_{k,\iota}\right)
\end{equation}
 For $\iota =1,\dots,4$, we have, using \eqref{minchar}, the identity
$$\sum_{j\neq k, j,k=1}^N  \Re a^*_{j,\iota}a_{k,\iota} =   \left|\sum_{ j=1}^Na_{j,\iota}\right|^2 -\sum_{j=1}^N a_{i,j}^2= - \sum_{j=1}^N a_{i,j}^2\ ,$$
and then since $\kappa_i=-\frac{1}{N-1}$ for $i=1,\dots,4$,
\begin{equation}\label{ioexAA}
  \sum_{\iota=1}^4 \left(\sum_{j=1}^N |a_{j,\iota}|^2 + \kappa_\iota \sum_{j\neq k, j,k=1}^N  \Re a^*_{j,\iota}a_{k,\iota}\right) = \frac{N}{N-1} \sum_{\iota=1}^4 \sum_{j=1}^N |a_{j,\iota}|^2\ .
\end{equation}

 For $\iota > 4$, we simply use the fact for  such $\iota$, $\kappa_\iota$ is ${\cal O}(1/N^2)$ or smaller, and this takes the place of \eqref{minchar}, which is not satisfied for such $\iota$, in eliminating a factor of $N$.  Then since the $N\times N$ matrix that has $0$ in every diagonal entry, and $1$ elsewhere has eigenvalues $N-1$ and $-1$, 
\begin{equation}\label{ioexAAA}
-\sum_{j=1}^N a_{j,i}^2 \leq  \sum_{j\neq k, j,k=1}^N  \Re a^*_{j,\iota}a_{k,\iota} \leq   (N-1) \sum_{j=1}^N a_{j,i}^2\ .
\end{equation}
Hence, for $\iota >4$, an upper bound on ${\displaystyle \left(\sum_{j=1}^N |a_{j,\iota}|^2 + \kappa_\iota \sum_{j\neq k, j,k=1}^N  \Re a^*_{j,\iota}a_{k,\iota}\right)}$ is
$$
\left(1 + \max\left\{ \frac{7N-3}{3(N-1)^4} \ ,\    (N-1)\frac{5N-3}{3(N-1)^3}  \right\} \right)\sum_{j=1}^N a_{j,i}^2
$$
where we have used \eqref{upperlow}. Evidently the maximum is furnished by the second quantitiy in the braces. Summing on $\iota>4$ and combining this with  \eqref{ioexAA} yields the upper bound in  \eqref{gp3}.

For $\iota>4$, a lower bound on ${\displaystyle \left(\sum_{j=1}^N |a_{j,\iota}|^2 + \kappa_\iota \sum_{j\neq k, j,k=1}^N  \Re a^*_{j,\iota}a_{k,\iota}\right)}$ is
$$
\left(1 - \max\left\{ (N-1)\frac{7N-3}{3(N-1)^4} \ ,\   \frac{5N-3}{3(N-1)^3}  \right\} \right)\sum_{j=1}^N a_{j,i}^2
$$
where we have again  used \eqref{upperlow}. Evidently the maximum is furnished by the first quantitiy in the  braces. Summing on $\iota>4$ and combining  this with  \eqref{ioexAA} yields the lower bound in  \eqref{gp3}.


By  Lemma~\ref{CGL}, $\min\{1-\kappa_\iota\} = \tfrac{5N-3}{3(N-1)^3}$, and thus we have the lower bound.  For the upper bound, we use $\kappa_\iota = -1/(N-1)$ for $\iota = 1,2,3,4$, and note that 
for all $N\geq 3$, $\frac{N}{N-1} \leq  \left(1 +\frac{5N-3}{3(N-1)^2}\right)$.  
This gives us the upper bound.  The rest is now clear, including the fact that \eqref{conlawcon} is the only source of non-uniqueness in the representation of $f\in {\mathcal A}_N$ in the form $f(\vec v) = \sum_{j=1}^N \varphi_j(v_j)$. 
\end{proof}

 There is another type of quantitative chaos estimate that we need. For any functions $\xi$ on $\R^3$ such that $\xi(v_k)\in L^2(\sigma_N)$ for some (and hence all) $k$, consider the conditional expectation
 \begin{equation}\label{condex}
 E\{ \xi(v_k) \ |\ v_j = v\}  = K\xi(v)
 \end{equation}
 for $j\neq k$. If the coordinate functions were exactly independent, this would simply be the expectation of $\xi(v_k)$, which is a finite constant. 
 It turns out that when $\xi(v_k) $ is a polynomial in $|v_k|^2$, the conditional expectation is at least bounded -- not only on $\ST$, which is trivial, but the bound is independent of $N$. 
 Here is one such estimate:
 
\begin{lm}\label{v8} For $\psi(v) = |v|^8$, there is a constant $C < \infty$ such that $\|K\psi\|_\infty \leq C$ for all $N$. 
\end{lm} 

\begin{proof}[Proof of Lemma~\ref{v8}]  The formula \eqref{factor} gives us 
\begin{eqnarray*}
K\psi(v) &=& \int_{{\mathcal S}_{N-1}}\left|  \sqrt{\frac{N -v^2}{N-1}} \vec y - \frac{1}{\sqrt{N(N-1)}} v \right|^8{\rm d}\sigma_{N-1} \\
&\leq& 2^7 \left(\left(\frac{N -v^2}{N-1}\right)^4 \int_{{\mathcal S}_{N-1}} |\vec y|^8 {\rm d}\sigma_{N-1} - \frac{|v|^8}{N^4(N-1)^4} \right)
\end{eqnarray*}
It is evident that $\int_{{\mathcal S}_{N-1}} |\vec y|^8 {\rm d}\sigma_{N-1}$ is bounded uniformly in $N$, and in fact,
$$\lim_{N\to\infty} \int_{{\mathcal S}_{N-1}} |\vec y|^8 {\rm d}\sigma_{N-1} = (2\pi/3)^{-3/2}\int_{\R^3} |y|^8 e^{-3|y|^2/2}\ .$$
\end{proof}

In the remainder of this section we collect the other estimates of this type that we need. Their proofs, which are more intricate but still largely computational, are presented in Appendix A.


\begin{lm}\label{K2lemA} There is a finite constant $C$ such that for all $N> 3$ and all $v$ such that $v=v_N$ for some $\vec v\in \ST$, 
\begin{equation}\label{v80}
|E\{ |v_1|^4 \ |\ v_N = v \}  - S(v)| \leq \frac{C}{N}
\end{equation}
where
\begin{equation}\label{v800}
S(v) =  \frac{N^2 + |v|^4 - 2N|v|^2  }{(N-1)^2}\ .
\end{equation}
\end{lm}

\begin{lm}\label{K2lem} There is a finite constant $C$ such that for all $N> 3$ and all $(v,w)$ such that $(v,w) = (v_{N-1},v_N)$ for some $\vec v\in \ST$, 
\begin{equation}\label{v80}
|E\{ |v_1|^4 \ |\ (v_{N-1},v_N) = (v,w) \}  - S(v,w)| \leq \frac{C}{N}
\end{equation}
where
\begin{equation}\label{v800}
S(v,w) =  \frac{N^2 + |v|^4 + | w|^4 + 2N|v|^2 + 2N|w|^2 + 2|v|^2|w|^2}{(N-2)^2}\ .
\end{equation}
\end{lm}

 \subsection{The operators $W^{(\alpha)}$ and $P^{(\alpha)}$}   

Let $\alpha\in [0,2]$, and define the self adjoint operator $P^{({\alpha})}$ by 
\begin{equation}\label{pgdefHH}
P^{({\alpha})} = \frac{1}{N}\sum_{k=1}^N w_N^{\alpha/2}(v_k)  P_k \ ,
\end{equation}
recalling that $w_N(v)$ is defined in \eqref{weightdef}, and $P_k$ is defined in \eqref{Pkfir0}, or equivalently \eqref{Pkfir}.
For each $k$, both $P_k$ and the multiplication operator 
${\displaystyle w_N^{\alpha/2}(v_k) }$ are commuting and self adjoint, and hence
$P^{({\alpha})}$  is indeed self adjoint, and non-negative. 
Since each $P_k$ is a projection,
\begin{equation}\label{pg2H}
\frac{1}{N}\sum_{k=1}^N \int_{\ST}w_N^{\alpha/2}(v_k) 
|P_kf|^2 \dd \sigma_N
 = \langle f, P^{({\alpha})} f\rangle_{L^2(\ST,\sigma_N)}\ .
 \end{equation}
 
 Define the function $W^{({\alpha})}$ by
 \begin{equation}\label{wgdef}
W^{({\alpha})} = \frac{1}{N}\sum_{k=1}^N w_N^{\alpha/2}(v_k) \ .
\end{equation}
Then 
\begin{equation}\label{wg2}
\frac{1}{N}\sum_{k=1}^N \int_{\ST}w_N^{\alpha/2}(v_k) 
f^2 \dd \sigma =  \int_{\ST} W^{({\alpha})} f^2\dd \sigma_N\ ,
 \end{equation}
and  we can write:
\begin{equation}\label{ddef}
\D(f,f) := \int_{\ST} W^{({\alpha})} f^2\dd \sigma_N
  - \langle f, P^{({\alpha})} f\rangle_{L^2(\ST,\sigma_N)}\ .
 \end{equation}

Equivalently, by the computations just below  \eqref{rec2H},
\begin{equation}\label{ddef2}
\D(f,f) = \frac{1}{N}\sum_{k=1}^N  \int_{\ST}w_N^{\alpha/2}(v_k) 
\left[f - P_k f\right]^2  \dd \sigma_N\ ,
 \end{equation}
and hence $\D(f,f) \geq  0$ for all $f$ since for each $k$, $|v_k|^2 \leq N-1$. (Recall that because of the momentum constrain, not all of the energy can reside in a single particle.) It follows that $\D(f,f) = 0$ if and only if 
 $f - P_k f = 0$ almost everywhere for each $k$, and then in this case
 $$\| f\|_2^2 - \langle f,P^{(0)}f\rangle_{L^2(\ST)} = {\mathcal D}_{N,0}(f,f) = \frac1N \sum_{k=1}^N \int_{\ST}|f - P_kf|^2{\rm d}\sigma_N  = 0\ .$$
 Evidently, $P^{(0)}$ is a contraction, and $1$ is an eigenvalue of multiplicity one, and the eigenspace is spanned by the constant function $1$ \cite{CGL}. 
 This proves:
 
 \begin{lm}\label{nonzer} For all $N\geq 2$ and all $\alpha\in [0,2]$, and all non-zero $f\in L^2(\ST)$ that are orthogonal to the constants, $\D(f,f) > 0$. 
 \end{lm}

We use the following lemma proved in \cite[Lemma 3.5]{CCL14}:
 
 \begin{lm}\label{comp1}
 For all $0 < \alpha \le 2$ and all $x > -1$,
 \begin{equation}\label{com}
 (1+x)^ {\alpha/2} \geq 1 +  \tfrac{\alpha}{2} x - (1- \tfrac{\alpha}{2} )x^2\ .
 \end{equation}
 \end{lm}

\begin{lm}\label{weight} For all $N$, all $0 < \alpha \leq 2$, and for all $\vec v \in \ST$,
\begin{multline}\label{wlb1}
1   - (1-\tfrac{\alpha}{2})\frac{N((N-1)^2 + 1)}{(N-1)^4} - \tfrac{\alpha}{2}\frac{1}{(N-1)^2}   + (1-\tfrac{\alpha}{2})\frac{N+1}{(N-1)^3} \ \leq \ W^{(\alpha)}(\vec v)\  \leq\\ 
\left(1 - \frac{1}{(N-1)^2}\right)^{\alpha/2}\ .
\end{multline}
Furthermore,  for all $\vec v \in \ST$, with $W^{(\alpha)}$ given by \eqref{wgdef},
\begin{equation}\label{wlb2}
W^{(0)}(\vec v)\     = \  1\qquad{\rm and}\qquad \  W^{(2)}(\vec v)\   = \ 1- \frac{1}{(N-1)^2}\ .
\end{equation}
\end{lm}

\begin{proof}
 Repeated use will be made of
 \begin{equation}\label{toten}
 \frac{1}{N}\sum_{k=1}^N|v_k|^2 =1 \;
 \end{equation}
 that identity  is part of  the definition of $\ST$. 

Because of \eqref{toten}, ${\displaystyle \frac{1}{N}\sum_{k=1}^N    \left(\frac{N^2 - (1+|v_k|^2)N}{(N-1)^2}\right) = 1- \frac{1}{(N-1)^2}}$.
Since $x \mapsto x^{\tfrac{\alpha}{2}}$ is concave on $\R_+$ for $0 \leq \alpha \leq 2$, 
Jensen's inequality yields the upper bound.

To prove the lower bound, use the inequality \eqref{com}:
Writing
\begin{equation}\label{rateredef}
\frac{N^2 - (1+|v_k|^2)N}{(N-1)^2}  = 1+ \frac{N(1-|v_k|^2) -1}{(N-1)^2}\ ,
\end{equation}
and applying \eqref{com} and \eqref{toten} yields
$$W^{({\alpha})}(\vec v) \geq  1 - \tfrac{\alpha}{2} \frac{1}{(N-1)^2} - (1-\tfrac{\alpha}{2})\frac1N \sum_{k=1}^N 
\left( \frac{N(1- |v_k|^2) -1}{(N-1)^2}\right)^2\ .
$$
Expanding the square on the right and applying \eqref{toten} twice more, we find
\begin{equation}\label{wabndX}
W^{({\alpha})}(\vec v) \geq  1 - \tfrac{\alpha}{2} \frac{1}{(N-1)^2}  - \frac{1-\tfrac{\alpha}{2}}{(N-1)^4}\left[ 1-N^2+ N \sum_{k=1}^N 
|v_k|^4\right] \ .
\end{equation}

The maximum of ${\displaystyle \sum_{k=1}^N |v_k|^4}$ on $\ST$  is no greater than the maximum of the convex function
$ \sum_{k=1}^N x_k^2$
on the convex set of $(x_1,\dots,x_N)$ satisfying
\begin{equation}\label{wlb4}
0 \leq x_j \leq N-1  \quad {\rm for\ all}\quad j =1,\dots,N \qquad{\rm and}\qquad \sum_{j=1}^Nx_j = N\ .
\end{equation}
The extreme points are obtained by permuting the coordinates of $(N-1,1,0,\dots,0)$. Evaluating the sum at such a point yields the stated bound,
The final statement is obvious.
\end{proof}

\begin{remark}\label{lowN} Lemma~\ref{weight} shows that for large $N$, 
\begin{equation}\label{wlowY1}
 W^{(\alpha)}(\vec v) \geq  1 - \left(1-\frac{\alpha}{2}\right) \frac{1}{N} + \mathcal{O}\left(\frac{1}{N^2}\right)\ . 
\end{equation}
The fact that the coefficient of $1/N$ is no less than $-1$ is essential for the result that we shall prove. 

We are particularly concerned with the case $\alpha =1$, and shall provide all the details in this case only. For $\alpha =1$, the lower bound simplifies further to 
\begin{equation}\label{a1ilow}
W^{(1)}(\vec v) \geq   C_N := 1   - \tfrac12 \frac{1}{N-1}  - \tfrac12 \frac{1}{(N-1)^2}   + \tfrac12 \frac{1}{(N-1)^3}    - \tfrac12 \frac{1}{(N-1)^4}   
\end{equation}

It is easily seen that for all $N\geq 2$,  $C_N$ increases as $N$ increases. For small $N$, we have the explicit values
$$C_3 = \frac{21}{32}\quad{\rm and} \quad C_4 = \frac{64}{81}\ .$$
\end{remark}

 We  now turn to $P^{(\alpha)}$.
By \eqref{rateredef}, for each $k$ for all $v_k$, 
\begin{equation}\label{ptws}
w_{N}^{\alpha/2}(v_k) =  \left[\frac{N^2 - (1+|v_k|^2)N}{(N-1)^2}\right]^{\alpha/2}  = (1+x_N(v_k))^{\alpha/2} \ ,
\end{equation} 
where
\begin{equation}\label{xNdef}
x_N(v_k) =\frac{1}{N-1} - \frac{N}{(N-1)^2}|v_k|^2\ .
\end{equation}
Note that 
${\displaystyle 
-1\leq x_N(v_k) \leq  \frac{1}{N-1}}$. 
Then
\begin{equation}\label{xNdef2}
w_N^{\alpha/2}(v_k)\leq \left(\frac{N}{N-1}\right)^{\alpha/2}
\end{equation}
and by  \eqref{ptws} and the  bounds from Lemma \ref{comp1},  $1 + \tfrac{\alpha}{2}x + (\frac{\alpha}{2} -1)x^2 \le (1+x)^{\alpha/2} \leq 1 + \tfrac{\alpha}{2}x$, 
\begin{equation}\label{xNdef1}
 |w_N^{\alpha/2}(v_k) -1| \leq  \frac{1}{N-1} + \frac{3N}{(N-1)^2}|v_k|^2 + \frac{N^2}{(N-1)^4}|v_k|^4
\end{equation}
 for all $\alpha\in [0,2]$, where we have made estimates to simplify the right hand side.  Thus, while $W^{(\alpha)}$ is only constant for $\alpha =0,2$, it is nearly constant for all $\alpha\in (0,2)$ when $N$ is large. However, its range, and hence the spectrum of the multiplication operator specified by $W^{(\alpha)}$, is a closed interval of positive length.
 At this point we record a simple lemma that will be useful later.
 
 \begin{lm}\label{Lplem}  For all $p \geq 1$,  there a constant $C$ depending only on $p$, so that for an $N\geq 3$ and all $\alpha\in [0,2]$,
 \begin{equation}\label{lpbndW}
\left( \int_{\ST}   |w_N^{\alpha/2}(v_k) -1|^p{\rm d}\sigma_N\right)^{1/p} \leq \frac{C}{N}\ .
 \end{equation}
 \end{lm}
 
 \begin{proof}This is an immediate consequnce of \eqref{xNdef1}, the triangle inequality, and the fact that for all $m\in \N$,
 $$\lim_{N\to\infty}\int_{\ST} |v_k|^{mp}{\rm d}\sigma_N =   \int_{\R^3} |v|^{mp}{\rm d}\gamma\ .$$
 \end{proof}

\begin{lm}\label{nullalpha} For all $\alpha \in [0,2]$, the null space of $P^{(\alpha)}$ is independent of $\alpha$.  If $h$ belongs to the null space of $P^{(0)}$, 
then $P_k h =0$ for each $k=1,\dots N$.  For all $\alpha \in [0,2]$, the closure of the range of $P^{(\alpha)}$ is the subspace 
${\mathcal A}_N$  of $L^2(\ST)$ defined in Definition~\ref{ANdef}. 
\end{lm}

\begin{proof} Since  $P^{(\alpha)}\geq 0$,  $h$ belongs to the null space of $P^{(\alpha)}$ if and only if $\langle h, P^{(\alpha)} h \rangle = 0$. But 
  ${\displaystyle 0 =  \langle h,  P^{({\alpha})} h\rangle
  = \frac{1}{N}\sum_{k=1}^N\int_{\ST}w_N^{\alpha/2}(v_k)    |P_k h|^2\dd \sigma_N}$.
Since $w_N^{\alpha/2}(v_k)    \ge 0$ almost everywhere,
it must be the case that $|P_k h|^2$ vanishes identically.  Thus is $h$ is in the null space of $P^{(\alpha)}$, $P_k h = 0$ for each $k$, and $h$ is in the null space of 
$P^{(0)}$. Conversely if $h$ is in the null space of $P^{(0)}$, then $P_k h = 0$ for each $k$, and then clearly $P^{(\alpha)}h = 0$. 

Since each $P^{(\alpha)}\geq 0$, the closure of its range is the orthogonal complement of its null space. Since the null space does not depend on $\alpha$, neither does the range. Evidently, ${\mathcal A}_N$ is the closure of the range of $P^{(0)}$. 
\end{proof}

\subsection{The spectrum of $\widehat{L}_{N,0}$}

Already in our paper \cite{CCL03} we have proved results that specify the exact spectral gap of  $\widehat{L}^{N,\alpha}$ for $\alpha = 0$. 
This case is especially amenable  for several reasons. First, 
since $W^{(0)} =1$,
$$\widehat{L}_{N,0} f = f - P^{({0})}f\ ,$$
and hence the problem is to determine the spectrum of $P^{({0})}$. Second,  there is an orthonormal basis of $L^2(\ST)$  
consisting of eigenfunctions of $P^{(0}$.  This is the case because each  $P_k$ is an average of rotations, so the finite dimensional spaces spanned by spherical harmonics of given maximal degree are invariant under 
$P^{(0)}$, and therefore one can study the spectrum of $P^{(0)}$ by studying the eigenvalue equation $P^{(0)}f = \lambda f$. This is the approach we took
in our previous work. However, this approach cannot work even  for $\alpha =2$, the next simplest case:  In this case, $P^{(2)}$ has an interval of continuous spectrum, as we shall see below. 
Therefore, we now give another argument that determines the spectral gap of $\widehat{L}_{N,0}$ that does extend to $\alpha =2$ at least. 

\begin{lm}  For all $N\geq 3$, 
\begin{equation}\label{congap0}
\widehat{\Delta}_{N,0} = 1 - \frac{3N -1}{3(N-1)^2} = 1 - \frac1N + \mathcal{O}\left(\frac{1}{N^2}\right)\ .
\end{equation}
and the second largest eigenvalue of $P^{(0)}$, $\mu^{(0)}$, is given by
\begin{equation}\label{mu0val}
\mu^{(0)} = \frac{3N -1}{3(N-1)^2}\ .
\end{equation}
\end{lm}

\begin{proof}The range of $P^{(0)}$ is ${\mathcal A}_N$, and it suffices to determine the spectrum of $P^{(0)}$ as an operator on ${\mathcal A}_N$. For 
${\displaystyle f(\vec v) = \sum_{j=1}^N \varphi_j(v_j)\in {\mathcal A}_N}$, we compute
$$P^{(0)} \left( \sum_{j=1}^N \varphi_j(v_j)\right) = \frac1N \sum_{k=1}^N \left(\varphi_k(v_k) + \sum_{j\neq k, j=1}^N K\varphi_j(v_k)\right)\ .$$

By  this computation, 
with $T: \bigoplus_{j=1}^N {\mathcal H}_{N,j} \to {\mathcal A}_N$ defined as in Theorem~\ref{gprop2},
$$ T^{-1}P^{(0)} T =   {\bf M}^{(0)}$$
where
 ${\bf M}^{(0)} = [M_{i,j}^{(0)}]$ is the 
$N\times N$ block matrix operator on   $\bigoplus_{j=1}^N {\mathcal H}_{N,j} $ given by 
$$M_{i,j}^{(0)} =  \frac{1}{N} I \quad{\rm if}\  i=j\quad{\rm and}\quad M_{i,j}^{(0)} = \frac{1}{N} K \quad{\rm if}\  i\neq j\ . $$

Note that ${\bf M}^{(0)}$ is  unitarily equivalent to the block matrix operator in $\bigoplus_{j=1}^N {\mathcal H}_{N,j} $ given by 
$$
\frac{1}{N} \left[\begin{array}{ccccc} I + (N-1)K & 0 & 0  & \cdots & 0\\
0 & I-K & 0 &\cdots & 0\\
0 & 0 &  I-K &\cdots & 0\\
\vdots & \vdots & \vdots & \ddots & \vdots \\
0 & 0 & 0 & 0 & I -K
\end{array}
\right]\ .
$$
It follows that  if $\lambda \not= 0$ is an eigenvalue of $P^{(0)}$ then either $\lambda$ is an eigenvalue of $\frac{1}{N}(I + (N-1)K)$ or else $\lambda$ is an eigenvalue of $\frac{1}{N}(I - K)$. Thus, the second largest eigenvalue of ${\bf M}^{(0)}$, and hence $P^{(0)}$, is 
either $1 + (N-1)\kappa$, where $\kappa$ is the second largest eigenvalue of $K$, or else
$1-\kappa$ where $\kappa$ is the least eigenvalue of $K$. From the information on the spectrum of $K$ provided in Lemma~\ref{CGL}, one immediately deduces \eqref{mu0val}, and then \eqref{congap0} follows directly. 
\end{proof}

\subsection{The spectrum of $\widehat{L}_{N,\alpha}$, $\alpha \in (0,2]$} 

\medskip  After $\alpha =0$, the next simplest case is $\alpha =2$ since then at least $W^{(2)}$ is constant; as we have seen
$W^{(2)} = 1 - (N-1)^{-2}$.
It follows that  $1$ is an eigenfunction for $P^{(2)}$ with eigenvalue $1 - (N-1)^{-2}$, and it spans the eigenspace. That is,
$1$ spans the null space of $\widehat{L}_{N,2}$. 

\begin{lm}\label{2yes} For all $N \geq 3$, $\widehat{\Delta}_{N,2} > 0$.
\end{lm}

\begin{proof}
The range of $P^{(2)}$ is ${\mathcal A}_N$, and as with $\alpha = 0$, $ {\bf M}^{(2)} := T^{-1}P^{(2)} T$ has  a simple block matrix structure:
\begin{eqnarray*}
P^{(2)}\left( \sum_{j=1}^N \varphi_j(v_j)\right) &=& \frac1N \sum_{k=1}^N w_{N,2}(v_k) P_k\left( \sum_{j=1}^N \varphi_j(v_j) \right)\\
&=& \sum_{k=1}^N \frac1N w_{N,2}(v_k) \left( \varphi_k(v_k) + \sum_{j\neq k, j=1}^N K\varphi_j(v_k) \right)
\end{eqnarray*}

By Theorem~\ref{gprop2}, ${\bf M}^{(2)} = T^{-1}P^{(2)} T = {\bf W}^{(2)}( {\bf I} + {\bf C})$, where
$$
{\bf W}^{(2)} = \frac1N \left[\begin{array}{cccc} 
w_{N,2}(v_1) & 0 & \cdots & 0\\
0 & w_{N,2}(v_2) & \cdots & 0\\
\vdots & \cdots & \ddots & \vdots\\
0 & \cdots & 0 &w_{N,N}(v_N)\end{array}\right]\ ,
$$
${\bf I}$ is the identity on $\bigoplus_{j=1}^N {\mathcal H}_{N,j}$, and ${\bf C}$ is given by
$$
{\bf C}=  \left[\begin{array}{ccccc} 
0 & K &  K &\cdots & K\\
K & 0 &  K &\cdots & K\\
\vdots & \cdots & \cdots &  \ddots & \vdots\\
K &  \cdots&  K & 0 &K\\
K &  \cdots &  K  & K & 0\end{array}\right]\ ,
$$

Since ${\bf M}^{(2)}$ and $P^{(2)}$ are similar, they have the same spectrum, and in particular, the spectrum of  ${\bf M}^{(2)}$ is real. 
(This is also evident from the identity  ${\bf M}^{(2)} = {\bf W}^{(2)}( {\bf I} + {\bf C})$, and the fact that for bounded operators 
$A$ and $B$ on any Hilbert space, $AB$ and $BA$ have the same spectrum.)

Since  the range of $\frac{1}{N} w_{N,2}$ is $[0,(N-1)^{-1}]$,
this interval is the spectrum of ${\bf W}^{(2)}$.   Note that ${\bf C}$, and hence ${\bf W}^{(2)}{\bf C}$ is compact. 
By Weyl's lemma, the essential spectrum of $T P^{(2)} T^{-1} $, and hence of $P^{(2)}$, is the essential spectrum of ${\bf W}^{(2)}$, which is the interval 
$[0,(N-1)^{-1}]$.  
 Hence any spectrum of $P^{(2)}$ in $(N-1)^{-1}, 1- (N-1)^{-2})$ consists of isolated eigenvalues,  and the isolated eigenvalues can 
 only accumulate 
 at a point in  $[0,(N-1)^{-1}]$.    In particular, $1 - (N-1)^{-2}$ cannot be an accumulation point, and hence $P^{(2)}$ 
 has a spectral gap below its top eigenvalue $1- (N-1)^{-2}$.  This proves that
$\widehat{\Delta}_{N,2} > 0$ for all $N\geq 3$. 
\end{proof}

For $\alpha \in (0,2)$, $W^{(\alpha)}$ is not constant -- although for large $N$ it is nearly constant. This means that for such $\alpha$,  
one cannot determine the spectrum of $\widehat{L}_{N,\alpha}$ simply by determining the spectrum of $P^{(\alpha)}$, and moreover, 
${\mathcal A}_N$ is not invariant under 
$\widehat{L}_{N,\alpha}$.   However, there is a simple comparison that one can make between $\D$ and ${\mathcal D}_{N,2}$ that provides the bound on 
$\widehat{\Delta}_{N,\alpha}$ that we seek. 

\begin{lm}\label{firpar}  For all $N\geq 3$, and all $\alpha\in [0,2]$, 
$$\widehat{\Delta}_{N,\alpha} \geq \left(\frac{N-1}{N}\right)^{1-\alpha/2}\widehat{\Delta}_{N,2}> 0\ .$$
\end{lm}

\begin{proof}
By \eqref{xNdef2}, for all $f$ and $k$, and all $\alpha \in (0,2)$,
$$\left[\frac{N^2 - (1+|v_k|^2)N}{(N-1)^2}\right]^{\alpha/2} 
\left[f - P_k f\right]^2   \geq \left(\frac{N-1}{N}\right)^{1-\alpha/2}
\left[\frac{N^2 - (1+|v_k|^2)N}{(N-1)^2}\right] 
\left[f - P_k f\right]^2  $$
It follows immediately that 
${\displaystyle \D(f,f) \geq   \left(\frac{N-1}{N}\right)^{1-\alpha/2}  {\mathcal D}_{N,2}(f,f)}$,
and then that
${\displaystyle\widehat{\Delta}_{N,\alpha} \geq   \left(\frac{N-1}{N}\right)^{1-\alpha/2} \widehat{\Delta}_{N,2} >0}.$
\end{proof}

At this point, we have proved the first part of  Theorem~\ref{thm2X}, and all that remains is to prove the second part.

\section{A sharper lower bound on $\widehat{\Delta}_{N,1}$  for large $N$}

In this section we obtain lower bounds on ${\mathcal D}_{N,1}(f,f)$ for $f$ orthogonal to the constants that become sharper and sharper as $N$ increases.  
To keep the computations simple, we do this explicitly for $\alpha =1$, though the method applies to all $\alpha\in (0,2)$. We shall prove the following, which is simply a specialization  of Theorem~\ref{thm2X}:

 \begin{thm}\label{Dmain}  There is a constant $C$ independent of $N$ such that whenever $f$ is orthogonal to the constants, 
 \begin{equation}\label{mainDbnd}
 {\mathcal D}_{N,1}(f,f) \geq \left( 1 - \frac1N - \frac{C}{N^{3/2}}\right)\|f\|_2^2\ .
 \end{equation}
 \end{thm} 
 
 The bound \eqref{mainDbnd} is meaningless for $N$ such that the right side is negative. However, no matter what $C$ is, there is an $N_0\in \N$ such that
 for all $N \geq N_0$, the right side is positive. From that point on, we have what we need for our induction. 
 The rest of this section is devoted to the proof of Theorem~\ref{Dmain}

 \subsection{The trial function decomposition}

We begin by specifying a trial function decomposition that we shall use. Let ${\mathcal A}_N$ be the subspace of $L^2(\sigma_N)$ defined in Definition~\ref{ANdef}. 
For any $f\in L^2(\ST)$ orthogonal to the constants, define $p$ and $h$ to be the orthogonal projections of $f$ onto ${\mathcal A}_N$ and ${\mathcal A}_N^\perp$ respectively. 
Then since $1\in {\mathcal A}_N$, $h$ is orthogonal to the constant, and then $p =f- h$ is orthogonal to the constants. 

By Lemma~\ref{nullalpha},  $h$ is  the component of $f$ in the null space of  $ P^{({\alpha})}$ for each $\alpha\in [0,2]$, and hence 
\begin{equation}\label{red2}
\langle f, P^{({\alpha})} f \rangle  = \langle p, P^{({\alpha})} p \rangle
\end{equation}
which yields

\begin{equation}\label{red2X}
  \int_{\ST} W^{({\alpha})} f^2\dd \sigma
  - \langle f, P^{({\alpha})} f\rangle_{L^2(\ST)} =
  \int_{\ST} W^{({\alpha})} f^2\dd \sigma
  - \langle p, P^{({\alpha})} p\rangle_{L^2(\ST)} \ .
 \end{equation}
  
  Since $p\in {\mathcal A}_N$, 
  there are $N$  functions $\phi_1,\dots,\phi_N$ of
a single variable such that  $\phi_j(v_j)\in L^2(\ST)$ for each $j$, and 
\begin{equation}\label{struc}
p(\vec v) = \sum_{j=1}^N\phi_j(v_j)\ ,
\end{equation}
and we shall always choose the particular representation of this form that is specificed in Theorem~\ref{gprop2}. That is, the eigenfunctions expansion
\begin{equation}\label{struc2}
\phi_j = \sum_{i=1}^\infty a_{j,i}\eta_i(v_j)
\end{equation}
given in \eqref{goodexp} is such that \eqref{minchar} is satsified; i.e., $\sum_{j=1}^N a_{j,i} = 0$ for $j=1,\dots,4$.
%
We make a further decomposition of $\phi_j(v_j)$ as follows:

\begin{defi}\label{decompdef}  Let  $p$ be  a function given by a sum of the form \eqref{struc} where for each $j$, $\phi_j(v_j)$ is orthogonal to the constants, and moreover, \eqref{struc2} is satisfied with 
$\sum_{j=1}^N a_{j,i} = 0$ for $j=1,\dots,4$.
Define  
\begin{equation}\label{struc3}
\psi_j = \sum_{i=1}^4 a_{j,i}\eta_i(v_j)\quad{\rm and} \quad\varphi_j(v_j) =  \sum_{i=5}^\infty a_{j,i}\eta_i(v_j)
\end{equation}
so that $\phi_j = \psi_j+ \varphi_j$. 
Next, define 
\begin{equation}\label{pdecomp}
g(\vec v) = \sum_{j=1}^N \varphi_j(v_j) \quad{\rm and}\quad s(\vec v) = \sum_{j=1}^N \psi_j(v_j)
\end{equation}

Finally the {\em trial function decomposition} of any $f\in L^2(\sigma_N)$ that is orthogonal to the constants is given by
\begin{equation}\label{pdecomp}
f = g+ s + h
\end{equation}
where $h$ is the component of $f$ in the null space of $P^{(\alpha)}$, $p$ is the component of $f$ in the closure of the range of $P^{(\alpha)}$, and 
$p =g+s$ is the decomposition of $p$ defined in \eqref{pdecomp}.
\end{defi}

\begin{remark} It is easy to see that when $p$ is symmetric under coordinate  permutations, one can take the functions $\phi_j$ in \eqref{struc} to be all the same. 
In particular, each $\psi_j$ has the form $\psi_j(v_j) = a\cdot v_j + b(|v_j|^2-1)$ for some fixed $a\in \R^3$ and $b\in \R$. Then
$$s(\vec v) = a\cdot\left(\sum_{j=1}^N v_j\right) + b\left(\sum_{j=1}^N (|v_j|^2-1)\right) = 0$$
on account of the constraints on the momentum and energy. Hence when $p$ is symmetric $s=0$, and in this case 
the trial function decomposition simplifies to $f = g+h$, as in \cite{CCL14}.
\end{remark}

We have seen in Lemma~\ref{nullalpha} that for each $k$, $P_kh = 0$. The next lemma shows that each $P_k$ also has a simple action on $s$:

\begin{lm}\label{Pks} For each $k$, the function $s$ in the trial function decomposition satisfies
\begin{equation}\label{pkse}
P_k s(\vec v) =  \frac{N-2}{N-1}\psi_k(v_k)\ .
\end{equation}
\end{lm}

\begin{proof} Note that
${\displaystyle
 P_k s(\vec v) =\psi_k(v_k) - \frac{1}{N-1}\sum_{j\neq k} \psi_j(v_k)}$.
Writing $\psi_j = \sum_{i=1}^4 a_{j,i}\eta_i$ and recalling  that  $\sum_{j=1}^N a_{j,i} =0$ for $i =1,\dots,4$, for any {\em fixed} $v$
$$\sum_{j=1}^N \psi_j(v) =  \sum_{j=1}^N \sum_{i=0}^4 a_{j,i}\eta_i(v) =  \sum_{i=0}^4 \left(\sum_{j=1}^Na_{j,i}\right)\eta_i(v) = 0 \ ,$$
from which \eqref{pkse} follows.  
\end{proof}

Each of the components $g$, $s$ and $h$ have their own special properties that we shall repeatedly use. 

\smallskip
\noindent{\it (1)}  A very useful feature of $g(\vec v) = \sum_{j=1}^N \varphi_j(v_j)$ is that, by Lemma~\ref{CGL}  for each $j$
$$\|K\varphi_j\|_2 \leq \frac{5N-3}{3(N-1)^3}\|\varphi_j\|_2 \ .$$
This gives us something {\em almost} like Lemma~\ref{Pks} for $g$:
$$P_k g(\vec v) = \varphi_k(v_k) + \sum_{j\neq k} K \varphi_j(v_k)\ ,$$
and hence
$$\|P_k g - \varphi_k(v_k)\|_2 \leq  \left\Vert \sum_{j\neq k} K \varphi_j(v_k)\right\Vert_2 \leq \frac{C}{N^2}\sum_k \Vert \varphi_k \Vert^2\ ,
$$
where in the last inequality we have used Theorem~\ref{gprop2}.

\smallskip
\noindent{\it (2)}  The key feature of $s(\vec v) = \sum_{j=1}^N \psi_j(v_j)$ is that  $P_k$ has a very simple action on  $s$, given in Lemma~\ref{Pks}.

Another is that each $\psi_j(v_j)$ belongs to 
$L^4(\sigma_N)$, and for a constant $C$ independent of $N$, $\|\psi_j\|_4 \leq C\|\psi_j\|_2$. This is essentially because  
the integrals $\int_{\ST}|v|^{2m}{\rm d}\sigma_N$ are bounded uniformly in $N$ for each $m$. In particular, if we wish to 
estimate the $L^2(\sigma_N)$ norm of $|v_k|^2\psi_k(v_k)$, we can apply Schwarz's inequality to bound this by $C\|\psi_k\|_4$, and then, 
changing $C$, to $C\|\psi_k\|_2$.   This will be used in estimating the quantity in \eqref{made1} below. 

\smallskip
\noindent{\it (3)}  A very useful feature of $h(\vec v)$ is that, by Lemma~\ref{nullalpha}, $P_k h = 0$ for each $k$, and in particular, $P^{(1)}h = 0$.

 \subsection{Lower bound on $ {\mathcal D}_{N,1}(f,f) $}
 
  For $\alpha=1$, the lower bound \eqref{wabndX} simplifies to
 \begin{equation}\label{wabnd1}
  W^{({1})}(\vec v) \geq \widetilde{ W}^{({1})}(\vec v) := 1 + \frac{1}{(N-1)^3}  -  \frac12\frac{N}{(N-1)^4} \sum_{k=1}^N 
|v_k|^4 \ .
\end{equation}
Define 
\begin{equation}\label{Dtodef}
\Dto (f,f) = \int_{\ST}  \widetilde{ W}^{({1})} f^2 {\rm d}\sigma_N - \langle f, P^{(1)} f\rangle\ .
\end{equation}
By \eqref{wabnd1}, $\Do (f,f)  \geq \Dto(f,f)$.

 Now let $f$ be orthogonal to the constants, and let $f = g+s + h$ be the trial function decomposition of $f$ as specified above. This notation will be used throughout this subsection. Note that
 
\begin{eqnarray*} 
 \Dto(f,f) 
  &=& \Dto(g,g) +  \Dto(s,s) + \Dto(h,h)\\
  &+& 2\Dto(g,h) +  2\Dto(s,h) + 2\Dto(g,s)
  \end{eqnarray*}
  
  The next lemma says that $g$, $s$ and $h$ are almost mutually orthogonal with respect to the inner product given by $\Dto$, and hence the 
  last three terms above make a negligible contribution.  This decouples the contributions of $g$, $s$ and $h$, which may then be analyzed 
  separately, taking advantage of their different helpful properties.
 
 \begin{lm}\label{almorth}  There is a constant $C$ independent of $N$ such that for any $f\in L^2(\sigma_N)$ that is orthogonal to the constants, 
 if $f = g+s+h$ is the trial function decomposition as specified above, then
 $$
 2|{\Dto}(g,h)|  +  2|{\Dto}(s,h)| + 2|{\Dto}(g,s)| \leq \frac{C}{N^{3/2}}\|f\|_2^2\ .
 $$
 \end{lm}

\begin{proof}
Since $P^{(1)} h = 0$, and since $g$ and $h$ are orthogonal, recalling that we may write $g(\vec v) = \sum_{j=1}^N \varphi_j(v_j)$,
\begin{eqnarray}
 \Dto(g,h)  &=& \int_{\ST} \widetilde W^{({1})} gh\dd \sigma_N =  -  \frac12\frac{N}{(N-1)^4}\sum_{k=1}^N \int_{\ST} |v_k|^4 g h {\rm d}\sigma_N\nonumber\\
&= &  -  \frac12\frac{N}{(N-1)^4} \sum_{k=1}^N   \int_{\ST} |v_k|^4 \varphi_k(v_k)h \dd \sigma_N\label{nct1}\\
&- & \frac12\frac{N}{(N-1)^4}   \sum_{j\neq k}^N   \int_{\ST}   |v_k|^4 \varphi_j(v_j) h \dd \sigma_N  \dd \sigma_N\label{nct2}
\end{eqnarray}
The integral in \eqref{nct1} vanishes since $P_k h =0$. 
Next consider the integral in \eqref{nct2}. 
 It will be convenient to introduce the notation $\xi(x) = x^8$
for the eighth power.  Then, with this definition, the Schwarz inequality, and then application of the $K$ operator, 
\begin{eqnarray}\label{red43}
\left|\int_{\ST}    |v_k|^4 \varphi_j(v_j) h\dd\sigma_N\right|
&\leq& \|h\|_2 \left(\int_{\ST} |v_k|^8  \varphi_j^2(v_j)\dd\sigma_N\right)^{1/2}\nonumber\\
&=&  \|h\|_2 \left(\int_{\ST} K\xi(v_j)  \varphi_j^2(v_j)\dd\sigma_N\right)^{1/2}\ .
\end{eqnarray}

By Lemma~\ref{v8}, there is a constant $C$ so  that, independent of $N$,  $\|K\xi\|_\infty \leq C$.
   Therefore,
${\displaystyle 
\left|\int_{\ST}  h  |v_k|^4 \varphi_j(v_j)\dd\sigma_N\right|
\leq C\|h\|_2 \|\varphi_j\|_2}$.
Using this in \eqref{nct2}  gives us
\begin{equation}
 \label{ghbound} \left| \frac{N}{(N-1)^4}  \int_{\ST} \left(\sum_{k=1}^N|v_k|^4\right) gh
 \dd \sigma \right| \leq  \frac{N}{(N-1)^3} C\|h\|_2\left(\sum_{j=1}^N\|\varphi_j\|_2\right)\ ,
 \end{equation}
and then since Theorem~\ref{gprop2} gives us
$$\sum_{j=1}^N\|\varphi\|_2   \leq \left(1 - \frac{5N-3}{3(N-1)^2}\right)^{-1/2}\sqrt{N}\|g\|_2\ ,$$
we have that the left hand side of (\ref{ghbound}) is bounded by
$\frac{C}{N^{3/2}}\|g\|_2\|h\|_2$ for a constant $C$ independent of $N$. 
We conclude that $|\Dto(s,h)| \leq CN^{-3/2}$.

 Finally, we consider $\Dto(s,g)$. This time we must also estimate $\langle s, P^{(1)} g\rangle$.  Because the span of $\{\eta_{j}(v_j)\ :\ 1 \leq j \leq N\}$ is invariant under $ P^{(0)}$, and every function in it is orthogonal to $g$,
$$
\langle s, P^{(1)} g\rangle =   \langle s, P^{(1)} g\rangle  - \langle P^{(0)}s,  g\rangle =  \langle s,( P^{(1)} - P^{(0)}) g\rangle\ .
$$
Introducing the short notation
$\tilde{w}(v)  := w_{N,1}(v)  -1$
{\em to be used in this proof only}, and writing
$$s(\vec v) = \sum_{\j=1}^N \psi_j(v_j) \quad{\rm and}\quad g(\vec v) = \sum_{j=1}^N \varphi_\ell(v_\ell)\ ,$$
we have  
\begin{equation}\label{partfo}
\langle s, P^{(1)} g\rangle =  \sum_{j,k,\ell=1}^N \int_{\ST} \psi_j(v_j) \tilde{w}(v_k) P_k \varphi_\ell(v_\ell)\ .
\end{equation}
We now split the sum over $j$, $k$ and $\ell$, into five parts 
$${\it (i)}\ j=\ell =k  \, \quad {\it (ii)}\ j\ne k, \ell = k   \, \quad {\it (iii)}\ j = k, \ell \ne k \quad{\rm and}\quad   {\it (iv)}\ j=\ell, \ell \ne k\ ,$$
and finally, ${\it (v)} \ j\ne \ell, \ell \ne k, k\neq j$

\begin{eqnarray}
\langle s, P^{(1)} g\rangle &=& \langle s,( P^{(1)} - P^{(0)}) g\rangle\\
&=& \frac1N \sum_{k=1}^N \langle  \tilde{w}(v_k) \psi_k(v_k),  \varphi_k(v_k)\rangle\label{made1}\\
&+& \frac1N \sum_{k=1}^N\sum_{j\neq k}^N \langle  \tilde{w}(v_k) \psi_j(v_j),  \varphi_k(v_k)\rangle\label{made2}\\
&+& \frac1N \sum_{k=1}^N\sum_{\ell\neq k}^N \langle   \tilde{w}(v_k)\psi_k(v_k), P_k \varphi_\ell(v_k)\rangle\label{made3}\\
&+& \frac1N\sum_{k=1}^N\sum_{\ell\neq k}^N \langle   \tilde{w}(v_k) \psi_\ell(v_k), P_k \varphi_\ell(v_k)\rangle\label{made4}\\
&+&\frac1N  \sum_{j\neq k,k\neq \ell,\ell\neq j}^N \langle   \tilde{w}(v_k) \psi_j(v_j), P_k \varphi_\ell(v_k)\rangle\label{made5}
\end{eqnarray}

We estimate \eqref{made1} as follows, using Lemma~\ref{Lplem} to bound $ \| \tilde{w}(v_k) \|_4$:
\begin{eqnarray*}
\frac1N \left| \sum_{k=1}^N \langle   \tilde{w}(v_k) \psi_k(v_k),  \varphi_k(v_k)\rangle\right| &\leq& \frac1N\sum_{k=1}^N  \| \tilde{w}(v_k) \|_4 \|\psi_k\|_4 \|\varphi_k\|_2
\leq \frac{C}{N^2} \sum_{k=1}^N  \|\psi_k\|_2 \|\varphi_k\|_2\\
&\leq&  \leq \frac{C}{N^2} (\|s\|_2^2 + \|g\|_2^2)\ .
\end{eqnarray*}

Since $P_k\psi_j(v_k) = -\frac{1}{N-1}\psi_j(v_k)$, the argument used to estimate  \eqref{made1} shows that the absolute value of the sum in 
\eqref{made2} is bounded above by
$$\frac{C}{N^3} \sum_{k=1}^N\sum_{j\neq k}^N\|\psi_j\|_2\|\varphi_k\|_2   \leq \frac{C}{N^2} (\|s\|_2^2 + \|g\|_2^2) \ , $$
as we found for \eqref{made1}.
Since for $k\neq j$, $\|P_k\varphi_j\|_2 \leq CN^{-2}\|\varphi_j\|_2$ (by \eqref{ktop}), 
the argument used to estimate  \eqref{made1} shows that the absolute value of the sum in 
\eqref{made3} is bounded above by
$$\frac{C}{N^4} \sum_{k=1}^N\sum_{j\neq k}^N\|\psi_k\|_2\|\varphi_j\|_2   \leq \frac{C}{N^3} (\|s\|_2^2 + \|g\|_2^2) \ , $$
even better than the previous bounds.  
Finally, for the terms in \eqref{made5}, 
\begin{eqnarray*}
|\langle  \tilde{w}(v_k)\psi_j(v_j), P_k \varphi_\ell(v_\ell)\rangle| &=&  |\langle P_k \varphi_\ell(v_\ell)  \tilde{w}(v_k), P_k\psi_j(v_j) \rangle|\\
&=& \frac{1}{N-1} |\langle P_k \varphi_\ell(v_\ell)  \tilde{w}(v_k),\psi_j(v_j) \rangle|\\
&\leq&  \frac{1}{N-1} \|K \varphi_\ell\|_2   \|\psi_j \|_2 \leq \frac{C}{N^3} \| \varphi_\ell\|_2   \|\psi_j \|_2\ ,
\end{eqnarray*}
where in the last inequality we have used Lemma~\ref{CGL}, and the fact that for each $j$, $\psi_j$ is an eigenfunction of $K$ considered as an operator on $L^2(\ST)$, with eigenvalue $-\frac{1}{N-1}$. 
Thus, 
$$\frac{C}{N^4} \sum_{j\neq k,k\neq \ell,\ell\neq j}^N \|\psi_j\|_2\|\varphi_\ell\|_2 \leq \frac{C}{N^2} (\|s\|_2^2 + \|g\|_2^2) \ .$$
This proves
$|\langle s, P^{(1)} g\rangle| \leq \frac{C}{N^2} (\|s\|_2^2 + \|g\|_2^2)$.
\end{proof}

We now turn to the  estimation of $\Dto(g,g)$  and $\Dto(s,s)$. 

\begin{lm}\label{gterbo1}  There is a constant $C$ independent of $N\geq 3$ such that for  all $g$ and $s$ as above,
\begin{equation}\label{gerbo1a}
\langle g, P^{({1})} g \rangle \leq 
\frac{1}{N}\sum_{k=1}^N \int_{\ST} 
\left[\frac{N^2 - (1+|v_k|^2)N}{(N-1)^2}\right]^{1/2}  
\varphi^2_k(v_k) \dd \sigma + \frac{C}{N^2}\|g\|_2^2\ .
\end{equation}
and
\begin{equation}\label{gerbo1b}
\langle s, P^{({1})} s \rangle \leq 
\frac{1}{N}\sum_{k=1}^N \int_{\ST} 
\left[\frac{N^2 - (1+|v_k|^2)N}{(N-1)^2}\right]^{1/2}  
\psi_k^2(v_k) \dd \sigma \ .
\end{equation}

\end{lm}

\begin{proof} Note first of all that $P_k g = \varphi_k(v_k) + \sum_{j\neq k}K\varphi_j(v_k)$, and thus
\begin{eqnarray}
\langle g, P^{({1})} g \rangle &=& 
\frac{1}{N}\sum_{k=1}^N \int_{\ST} 
\left[\frac{N^2 - (1+|v_k|^2)N}{(N-1)^2}\right]^{1/2} |P_k g|^2\dd \sigma_N \nonumber\\
&=&\frac{1}{N}\sum_{k=1}^N \int_{\ST} 
\left[\frac{N^2 - (1+|v_k|^2)N}{(N-1)^2}\right]^{1/2}  
\varphi^2_k(v_k) \dd \sigma_N\nonumber\\
&+&
\frac{2}{N}\sum_{k=1}^N \sum_{j\neq k}\int_{\ST} 
\left[\frac{N^2 - (1+|v_k|^2)N}{(N-1)^2}\right]^{1/2} 
\varphi_k(v_k)K\varphi_j(v_k)\dd \sigma_N\label{min1}\\
&+&
\frac{1}{N}\sum_{k=1}^N \sum_{j\neq k}\sum_{\ell\neq k} \int_{\ST} 
\left[\frac{N^2 - (1+|v_k|^2)N}{(N-1)^2}\right]^{1/2} 
K\varphi_j(v_k) K\varphi_\ell(v_k)\dd \sigma_N\label{min2}
\end{eqnarray}

By the Schwarz inequality and \eqref{ktop}, the sum of integrals in (\ref{min1}) is bounded above by 
$$\frac{C}{N^3}  \sum_{k=1}^N \sum_{j\neq k} \|\varphi_j\|_2\|\varphi_k\|_2  \leq \frac{C}{N^2}\|g\|_2^2\ .$$

Similarly, by \eqref{ktop} and Lemma~\ref{gprop2}, the sum of integrals in (\ref{min2}) is bounded above by 
$$\frac{C}{N^4} \sum_{j,k=1}^N  \|\varphi_j\|_2\|\varphi_k\|_2  \leq \frac{C}{N^3}\|g\|_2^2 \ .$$
Using the two bounds we have just derived on (\ref{min1}) and  (\ref{min2}) respectively, yields \eqref{gerbo1a}.

$$
\langle s, P^{({1})} s \rangle = 
\frac{1}{N}\sum_{k=1}^N \int_{\ST} 
\left[\frac{N^2 - (1+|v_k|^2)N}{(N-1)^2}\right]^{1/2} |P_k s|^2\dd \sigma_N \ ,
$$
\eqref{gerbo1b} follows directly from Lemma~\ref{Pks}.
\end{proof}

\begin{lm}\label{gv4lem}  There is a constant $C$ such that for all $N$ and all $g$ and $s$ as above, 
\begin{equation}\label{gv4}
 \int_{\ST}  \sum_{k=1}^N |v_k|^4 g^2{\rm d}\sigma_N  \leq  \sum_{k=1}^N \int_{\ST}  \varphi_k(v_k)^2|v_k|^4 {\rm d}\sigma + {C}{N}\|g\|_2^2\ ,
 \end{equation}
 and
 \begin{equation}\label{gv4B}
\int_{\ST}   \sum_{k=1}^N |v_k|^4 s^2{\rm d}\sigma_N \leq  \sum_{k=1}^N \int_{\ST}  \psi_k(v_k)^2|v_k|^4 {\rm d}\sigma + {C}{N}\|s\|_2^2\ ,
 \end{equation}
\end{lm}
\begin{proof}
\begin{eqnarray}
\int_{\ST} \sum_{k=1}^N |v_k|^4 g^2 {\rm d}\sigma &=& \sum_{i,j,k=1}^N \int_{ST}  \varphi_i(v_i)\varphi_j(v_j)|v_k|^4 {\rm d}\sigma_N\nonumber\\
&=&  \sum_{k=1}^N \int_{\ST}  \varphi_k(v_k)^2|v_k|^4 {\rm d}\sigma\label{chi1}\\
&+& 2\sum_{k=1}^N\sum_{j\neq k} \int_{\ST}  \varphi_j(v_j)\varphi_k(v_k)|v_k|^4 {\rm d}\sigma\label{chi2}\\
&+& \sum_{k=1}^N\sum_{j\neq k} \int_{\ST}  \varphi_j(v_j)^2|v_k|^4 {\rm d}\sigma\label{chi3}\\
&+& \sum_{i\neq j, j\neq k, k\neq i} \int_{\ST} \varphi_i(v_i)\varphi_j(v_j)|v_k|^4 {\rm d}\sigma\label{chi4}
\end{eqnarray}
By Lemma~\ref{gv4lemX1}, Lemma~\ref{gv4lemX2} and Lemma~\ref{gv4lemX3}  below the terms in \eqref{chi2}, \eqref{chi3} and \eqref{chi4}
add up to no more than $CN\|g\|_2^2$, which proves \eqref{gv4}. The same argument using the same lemmas proves  \eqref{gv4B}.
\end{proof}

\begin{lm}\label{gv4lemX1}  There is a constant $C$ such that for all $N$ and all $g$  and $s$ as above, 
\begin{equation}\label{gv4X1E}
2\sum_{k=1}^N\sum_{j\neq k} \int_{\ST}  \varphi_j(v_j)\varphi_k(v_k)|v_k|^4 {\rm d}\sigma \leq {C}{N}\|g\|_2^2\ .
\end{equation}
and
\begin{equation}\label{gv4X1E2}
2\sum_{k=1}^N\sum_{j\neq k} \int_{\ST}  \psi_j(v_j)\psi_k(v_k)|v_k|^4 {\rm d}\sigma \leq {C}\|s\|_2^2\ .
\end{equation}
\end{lm}

\begin{proof} For $j\neq k$,  using the pointwise bound $|v_k|^4 \leq (N-1)^2$ and then \eqref{ktop},
\begin{equation}\label{wrap1}
 \int_{\ST}  \varphi_j(v_j)\varphi_k(v_k)|v_k|^4 {\rm d}\sigma  \leq (N-1)^2\|K\varphi_j\|_2 \|\varphi_k\|_2 \leq \frac{(5N-3)(N-1)^2}{3(N-1)^3}\|\varphi_j\|_2 \|\varphi_k\|_2\ .
\end{equation}Then by Theorem~\ref{gprop2}, \eqref{gv4X1E} follows.
Next, 
\begin{equation}\label{wrap1}
 \int_{\ST}  \psi_j(v_j)\psi_k(v_k)|v_k|^4 {\rm d}\sigma  \leq \|K\psi_j\|_2 \||v_k|^4\psi_k\|_2 \leq  \frac{1}{N-1}\|\psi_j\|_2 C\|\psi_k\|_4 \leq \frac{C}{N}\|\psi_j\|_2 C\|\psi_k\|_2 \ .
\end{equation}Then by Theorem~\ref{gprop2} again, \eqref{gv4X1E2} follows.
\end{proof}

\begin{lm}\label{gv4lemX2}  There is a constant $C$ such that for all $N$ and all $g$ and $s$ as above, 
\begin{equation}\label{gv4X2E3}
\sum_{k=1}^N\sum_{j\neq k} \int_{\ST}  \varphi_j(v_j)^2|v_k|^4 {\rm d}\sigma_N  \leq CN\|g\|_2^2\ .
\end{equation}
and
\begin{equation}\label{gv4X2E3B}
\sum_{k=1}^N\sum_{j\neq k} \int_{\ST}  \psi_j(v_j)^2|v_k|^4 {\rm d}\sigma_N  \leq CN\|s\|_2^2\ .
\end{equation}
\end{lm}

\begin{proof}By Lemma~\ref{K2lemA}, there is a finite constant $C$ independent of $N$  such that 
$$
\sum_{k=1}^N\sum_{j\neq k} \int_{\ST}  \varphi_j(v_j)^2|v_k|^4 {\rm d}\sigma_N  \leq 
N\sum_{j=1}^N \int_{\ST}  \frac{N^2 + |v_j|^4 - 2N|v_j|^2  }{(N-1)^2}\varphi^2_j(v_j) {\rm d}\sigma_N + CN\sum_{j=1}^N \|\varphi_j\|_2^2\ .
$$
Then by Lemma~\ref{gprop2}, \eqref{gv4X2E3} follows.  The same analysis yields 
\eqref{gv4X2E3B} .
\end{proof}

\begin{lm}\label{gv4lemX3}  There is a constant $C$ such that for all $N$ and all $g$ and $s$ as above, 
\begin{equation}\label{gv4X3F}
\sum_{i\neq j, j\neq k, k\neq i} \int_{\ST} \varphi_i(v_i)\varphi_j(v_j)|v_k|^4 {\rm d}\sigma    \leq C\|g\|_2^2\ .
\end{equation}
and
\begin{equation}\label{gv4X3FB}
\sum_{i\neq j, j\neq k, k\neq i} \int_{\ST} \psi_i(v_i)\psi_j(v_j)|v_k|^4 {\rm d}\sigma    \leq C\|s\|_2^2\ .
\end{equation}

\end{lm}

\begin{proof}
By Lemma~\ref{K2lem}, there is a finite constant $C$ independent of $N$ (but changing from line to line) such that 
\begin{eqnarray*}
&&\sum_{i\neq j, j\neq k, k\neq i} \int_{\ST} \varphi_i(v_i)\varphi_j(v_j)|v_k|^4 {\rm d}\sigma    \leq  \frac{C}{N}\sum_{i \not=j}\|\varphi_i\|_2 \|\varphi_j\|_2+\\
&& (N-2)\sum_{i\neq j} \int_{\ST} \frac{N^2 + |v_i|^4 + |v_j|^4 + 2N|v_i|^2 + 2N|v_j|^2 + 2|v_i|^2|v_j|^2}{(N-2)^2}\varphi_i(v_i) \varphi_j(v_j) {\rm d}\sigma_N \ .
\end{eqnarray*}
Note that for $i\neq j$,
$$
\int_{\ST} \frac{N^2 + |v_i|^4 + |v_j|^4 + 2N|v_i|^2 + 2N|v_j|^2}{N-2}\varphi_i(v_i) \varphi_j(v_j) {\rm d}\sigma_N  \leq \frac{C}{N}
\|\varphi_i\|_2\|\varphi_j\|_2\ 
$$
since in each term we may  either replace $\varphi_i$ by $K\varphi_i$ or $\varphi_j$ by $K\varphi_j$, and this gives a factor of $CN^{-2}$. 
Then by Lemma~\ref{gprop2},  $\sum_{i\neq j}\|\varphi_i\|_2\|\varphi_j\|_2 \leq CN\|g\|_2^2$. 

The remaining terms must be handled differently. For $j=1,\dots,N$, let $\xi_j$ denote the function $\xi_j(v) = |v_j|^2\varphi_j(v_j)$, and note that $\xi_j$ is orthogonal to the constants. Therefore, 
\begin{eqnarray*}
\int_{\ST}  \frac{|v_i|^2|v_j|^2}{N-2}\varphi_i(v_i) \varphi_j(v_j) {\rm d}\sigma_N &=& \frac{1}{N-2}\langle \xi_i, K \xi_j \rangle\\
 &\leq& \frac{1}{N-2}\frac{1}{N-1}\|\xi_i\|_2\|\xi_j\|_2  \leq C\|\varphi_i\|_2\|\varphi_j\|_2
\end{eqnarray*}
Then by Lemma~\ref{gprop2},  $\sum_{i\neq j}\|\varphi_i\|_2\|\varphi_j\|_2 \leq CN\|g\|_2^2$, and \eqref{gv4X3F} follows.

Next, 
\begin{eqnarray*}
&&\sum_{i\neq j, j\neq k, k\neq i} \int_{\ST} \psi_i(v_i)\psi_j(v_j)|v_k|^4 {\rm d}\sigma    \leq   CN^2\|\psi_i\|_2 \|\psi_j\|_2 +\\
&& (N-2)\sum_{i\neq j} \int_{\ST} \frac{N^2 + |v_i|^4 + |v_j|^4 + 2N|v_i|^2 + 2N|v_j|^2 + 2|v_i|^2|v_j|^2}{(N-2)^2}\psi_i(v_i) \psi_j(v_j) {\rm d}\sigma_N \ .
\end{eqnarray*}
The main term is 
$$\frac{N^2}{N-2} \sum_{i\neq j} \int_{\ST}  \psi_i(v_i) \psi_j(v_j) {\rm d}\sigma_N = \frac{N^2}{N-2}\sum_{i\neq j}   \langle \psi_1, K\psi_j\rangle  
= \frac{N^2}{(N-2)(N-1)}\sum_{i\neq j} \|\psi_i\|_2 \|\psi_j\|_2\ , $$
and simple estimates show that all remaining terms are smaller.  \end{proof}

\subsection{Lower bounds on $\Dto(g,g)$ and  $\Dto(s,s)$}\label{sec3.3}

We are now ready to estimate $\Dto(g,g)$ and  $\Dto(s,s)$.  We first define a quadratic form ${\mathcal F}$ on $L^2(\sigma_N)$ as follows: For all functions $r$ in $L^2(\sigma_N)$, define
\begin{equation}\label{mvFdef}
{\mathcal F}(r,r)  :=  
\frac12\frac{N}{(N-1)^4}  \int_{\ST} \sum_{k=1}^N |v_k|^4 r^2
 \dd \sigma + \langle r, P^{({1})} r \rangle
\end{equation}

\begin{lm}\label{FFF}  For all $g$ and $s$ as above, 
\begin{equation}\label{FFF1}
\Dto(g,g) \geq \|g\|_2^2 - {\mathcal F}(g,g) \quad{\rm and}\quad \Dto(s,s) \geq \|s\|_2^2 - {\mathcal F}(s,s)\ .
\end{equation}
\end{lm}

\begin{proof} This is immediate from \eqref{wabndX}, \eqref{Dtodef} and the definition of ${\mathcal F}$. 
\end{proof}

\begin{lm} \label{FFG}
There is a finite constant $C$ independent of $N$ such that for all $g$ and $s$ as above, with ${\mathcal F}$ defined by \eqref{mvFdef}
\begin{equation}\label{FFbndg}
{\mathcal F}(g,g)   \leq \left( \frac1N + \frac{C}{N^2}\right)\|g\|_2^2\qquad{\rm and}\qquad 
{\mathcal F}(s,s)   \leq \left( \frac1N + \frac{C}{N^2}\right)\|s\|_2^2\ .
\end{equation}
\end{lm}

\begin{proof} By Lemma~\ref{gterbo1}  and Lemma~\ref{gv4lem},
\begin{eqnarray*}{\mathcal F}(g,g)    &\leq& \frac{1}{N}\sum_{k=1}^N \left[\int_{\ST} 
\left[\frac{N^2 - (1+|v_k|^2)N}{(N-1)^2}\right]^{1/2}  
 + \frac12\frac{N^2}{(N-1)^4} \sum_{k=1}^N \int_{\ST} |v_k|^4\right] \varphi(v_k)^2 {\rm d}\sigma\\
 &+&\frac{C}{N^2}\|g\|_2^2
\end{eqnarray*}
Define 
${\displaystyle y_k := \frac{N}{(N-1)^2}|v_k|^2}.$
Then $0 \leq y_k \leq N/(N-1)$, and 
\begin{equation}\label{Fbnd}
{\mathcal F}(g,g)    \leq  \frac{1}{N}\sum_{k=1}^N  \int_{\ST}  w(y_k) \varphi^2 (v_k){\rm d}\sigma_N  + \frac{C}{N^2}\|g\|_2^2\ 
\end{equation}
where
${\displaystyle w(y) =   \left(\frac{N}{N-1} - y\right)^{1/2} + \frac12 y^2}$.
Simple calculations show that $w(y) \leq \sqrt{N/(N-1)}$ for all $0 \leq y \leq N/(N-1)$, and in fact, for $N\geq 7$, $w(y)$ is monotone decreasing on this interval. 
Then 
\eqref{Fbnd} becomes  
$$
{\mathcal F}(g,g) \leq \sqrt{N/(N-1)} \frac{1}{N}\sum_{k=1}^N \|\varphi_k\|_2^2 + \frac{C}{N^2}\|g\|_2^2
$$
Now \eqref{FFbndg} follows directly from Theorem~\ref{gprop2}.   The proof of \eqref{FFbndg} is the same.
\end{proof}

\subsection{Proof of Theorem~\ref{Dmain}}

\begin{proof}[Proof of Theorem~\ref{Dmain}]
By Lemma~\ref{almorth},
$$\Do(f,f) \geq \Dto(f,f) \geq \Dto(g,g) + \Dto(s,s) + \Dto(h,h) - CN^{-3/2}\Vert f \Vert^2s\ .$$
By Lemma~\ref{FFF} and Lemma~\ref{FFG}, 
$$\Dto(g,g) + \Dto(s,s) \geq \left(1 - \frac{1}{N}  - \frac{C}{N^2}\right)(\|g\|_2^2 + \|s\|_2^2)\ .$$
Since $P^{(1)}h =0$,  \eqref{wlowY1} yields
${\displaystyle \Dto(h,h) \geq \left(1 - \frac{1}{2N} - \frac{C}{N^2}\right)\|h\|_2^2}$,
adding the estimates completes the proof since $\|f\|_2^2 = \|g\|_2^2+\|s\|_2^2+\|h\|_2^2$. 
\end{proof}

\begin{appendices}
\section{Some computational proofs}

\begin{proof}[Proof of Lemma~\ref{K2lemA}] By  \eqref{factor}, 
\begin{multline}\label{v44}
E\{ |v_1|^4 \ |\ v_N = v \}   =\\
\int_{{\mathcal S}_{N-1}} \left(\eta^4(v)|\vec y|^4 + \frac{\eta^2(v)}{(N-1)^2} | \vec y \cdot v|^2  + \frac{|v|^4}{(N-1)^4}
+ 2\frac{\eta^2(v)}{(N-1)^2}|\vec y|^2\right){\rm d}\sigma_{N-1}\ ,
\end{multline}
where
$$\eta^2(v,w) = \frac{N - |v|^2  - |v|^2/(N-1)}{N-1}\ ,$$
Define $M_N:= 
\int_{{\mathcal S}_{N-2}}|\vec y|^4 {\rm d}\sigma_{N-2}$ which is bounded uniformly in $N$:
$$\lim_{N\to\infty} \int_{{\mathcal S}_{N-1}} |\vec y|^4 {\rm d}\sigma_{N-1} = (2\pi/3)^{-3/2}\int_{\R^3} |y|^4 e^{-3|y|^2/2}\ .$$
Then the right hand side of \eqref{v44} becomes
\begin{equation}\label{v82X} M_N\eta^4(v) + \frac{1}{3(N-1)^2} \eta^2(v)|v|^2 +\frac{|v|^4}{(N-1)^4}
+ 2\frac{\eta^2(v)}{(N-1)^2}
\end{equation}
Note that for some constant $C$ independent of $N$, 
\begin{equation}\label{v83X} 
\frac{1}{3(N-1)^2} \eta^2(v)|v|^2 +
\frac{|v|^4}{(N-1)^4}
+ 2\frac{\eta^2(v)}{(N-1)^2} \leq \frac{C}{N}\ .
\end{equation}
Next, 
\begin{eqnarray}\label{v83} 
\eta^4(v)  = \frac{N^2 + |v|^4 - 2N|v|^2  }{(N-1)^2} 
+ \frac{|v|^4}{(N-1)^4} + 2\frac{(N-|v|^2)|v|^2}{(N-1)^3}\ .\nonumber
\end{eqnarray}
Again, for some constant $C$ independent of $N$, 
$$\frac{|v|^4}{(N-1)^4} + 2\frac{(N-|v|^2)|v|^2}{(N-1)^3} \leq \frac{C}{N}\ .$$
\end{proof}

\begin{proof}[Proof of Lemma~\ref{K2lem}] By a simple adaptation of \eqref{factor}, 
\begin{multline}\label{v8R}E\{ |v_1|^4 \ |\ (v_{N-1},v_N) = (v,w) \}  =\\ \int_{{\mathcal S}_{N-2}}\left(\beta^4(v,w)|\vec y|^4 + \frac{\beta^2(v,w)}{(N-2)^2} | \vec y \cdot (v+w)|^2  + \frac{|v+w|^4}{(N-2)^4}
+ 2\frac{\beta^2(v,w)}{(N-2)^2}|\vec y|^2\right){\rm d}\sigma_{N-2}\ ,
\end{multline}
where
$$\beta^2(v,w) = \frac{N - |v|^2 - |w|^2 - |v+w|^2/(N-2)}{N-2}\ ,$$
which is non-negative on the allowed values for $(v,w)$. Note that $\beta^2(v,w) \leq N/(N-2)$. Define $M_N:= 
\int_{{\mathcal S}_{N-2}}|\vec y|^4 {\rm d}\sigma_{N-2}$ which is bounded uniformly in $N$:
$$\lim_{N\to\infty} \int_{{\mathcal S}_{N-2}} |\vec y|^4 {\rm d}\sigma_{N-2} = (2\pi/3)^{-3/2}\int_{\R^3} |y|^4 e^{-3|y|^2/2}\ .$$
Then the right hand side of \eqref{v8R} becomes
\begin{equation}\label{v82} M_N\beta^4(v,w) + \frac{1}{3(N-2)^2} \beta^2(v,w)|v+w|^2 +\frac{|v+w|^4}{(N-2)^4}
+ 2\frac{\beta^2(v,w)}{(N-2)^2}
\end{equation}
Note that for some constant $C$ independent of $N$, 
\begin{equation}\label{v83}  \frac{1}{3(N-2)^2} \beta^2(v,w)|v+w|^2 +\frac{|v+w|^4}{(N-2)^4}
+ 2\frac{\beta^2(v,w)}{(N-2)^2} \leq \frac{C}{N}\ .
\end{equation}
Next, 
\begin{eqnarray}\label{v83} 
\beta^4(v,w)  &=& \frac{N^2 + |v|^4 + |w|^4 + 2N|v|^2 + 2N|w|^2 + 2|v|^2|w|^2}{(N-2)^2} \\
&+& \frac{|v+w|^4}{(N-2)^4} + 2\frac{(N-|v|^2-|w|^2)|v+w|^2}{(N-2)^3}\ .\nonumber
\end{eqnarray}
Again, for some constant $C$ independent of $N$, 
$$\frac{|v+w|^4}{(N-2)^4} + 2\frac{(N-|v|^2-|w|^2)|v+w|^2}{(N-2)^3} \leq \frac{C}{N}\ .$$
\end{proof}

\section{Quantitiative estimates on $\widehat{\Delta}_{N,2}$}

\subsection{An explicit bound for $N\geq 4$}

By \eqref{xNdef2},
$P^{(\alpha)}$ defined in \eqref{pgdefHH} satisfies
\begin{equation}\label{pabnd}
0 \leq  P^{(\alpha)}  \leq   \left(\frac{N}{N-1}\right)^{\alpha/2} P^{(0)} \ 
\end{equation}
for all $\alpha\in [0,2]$. As we have seen, 
the  second largest eigenvalue of   $P^{(0)}$, denoted $\mu^{(0)}_N$,
is given by 
\begin{equation}\label{mu0f}
\mu^{(0)}_N   = \frac{3N -1}{3(N-1)^2}\ .
\end{equation} 
It follows from \eqref{pabnd} and \eqref{mu0f} that  for all $f$ orthogonal to the constants,
\begin{equation}\label{mu1f}
\langle f,  P^{(\alpha)}f\rangle  \leq  \left(\frac{N}{N-1}\right)^{\alpha/2} \frac{3N -1}{3(N-1)^2}\|f\|_2^2\ ,
\end{equation}
for all $\alpha\in [0,2]$.   
For $\alpha = 2$, we have
${\displaystyle 
\langle f,  P^{(2)}f\rangle  \leq   \frac{N(3N -1)}{3(N-1)^3}\|f\|_2^2}$.
Note that 
$$
\frac{N(3N -1)}{3(N-1)^3} = \frac{1}{N-1} + \frac53 \frac{1}{(N-1)^2} + \frac23 \frac{1}{(N-1)^3}\ ,
$$
which evidently decreases monotonically as $N$ increases.  Next, since $W^{(2)}(\vec v) = 1 - \frac{1}{(N-1)^2}$, we have that for all $f$ orthogonal to the constants
\begin{equation}\label{aus5}
-\langle f, \widehat{L}_{N,2} f\rangle  = \langle f, (W^{(2)} - P^{(2)})f\rangle  \geq \left( 1 -  \frac{1}{N-1} - \frac83 \frac{1}{(N-1)^2} - \frac23 \frac{1}{(N-1)^3}\  \right)\|f\|_2^2\ .
\end{equation}
For $N= 3$, this yields only the useless bound   $-\langle f, \widehat{L}_{3,2} f\rangle \geq -\frac14 \|f\|_2^2$.  But already for $N=4$, it yields
$$-\langle f, \widehat{L}_{4,2} f\rangle \geq \frac{28}{81} \|f\|_2^2\ .$$  
Since the right hand side of \eqref{aus5} increases as $N$ increases, this, together with the comparison from Lemma~\ref{firpar}, proves:

\begin{thm} For all $N\geq 4$, 
$$\widehat{\Delta}_{N,2} \geq   1 -  \frac{1}{N-1} - \frac83 \frac{1}{(N-1)^2} - \frac23 \frac{1}{(N-1)^3} > 0\ ,$$
and for all $\alpha \in (0,2)$,
$$\widehat{\Delta}_{N,\alpha} \geq \left(\frac{N-1}{N}\right)^{1-\alpha/2}
\left( 1 -  \frac{1}{N-1} - \frac83 \frac{1}{(N-1)^2} - \frac23 \frac{1}{(N-1)^3}\  \right) > 0\ . $$
\end{thm}

At this point, the only estimate we lack for a fully quantitative  result is a quantitative estimate on $\widehat{\Delta}_{3,2}$.

\subsection{An explicit bound for $N=3$}

By what has been explained earlier, $\widehat{\Delta}_{3,2} = \frac34 - \nu_3$ where
\begin{equation}\label{nundefB}
\nu_3 = \sup\left\{ \langle f, P^{(2)} f\rangle_{L^2({\mathcal S}_3)}\ :\ \|f\|_2 =1\ ,\quad \langle f,1\rangle_{L^2(\ST)} = 0 \ \right\}\ ,
\end{equation}
and by Lemma~\ref{2yes}, $\widehat{\Delta}_{3,2}>0$, or, what is the same $\nu_3 < \frac34$. 

If $\nu_3 \leq \frac12$, then evidently $\widehat{\Delta}_{3,2} \geq \frac14$. Therefore, we need only consider the  possibility that $\nu_3 > \frac12$, and as we have seen, in 
in this case $\nu_3$ is an eigenvalue of $P^{(2)}$, and necessarily $\nu_3 < \frac 34$. 

In seeking the second largest eigenvalue of $P^{(2)}$, we need only consider functions $f$ of the form 
\begin{equation}\label{type1}
f(\vec v) = \sum_{j=1}^N \varphi(v_j)
\end{equation}
or 
\begin{equation}\label{type2}
f(\vec v) = \varphi(v_1)- \varphi(v_2)\ ,
\end{equation}
where in the second case  we have taken advantage of the the symmetry of $P^{(2)}$ to assume without loss of generality that $f$ is antisymmetric under interchange of $v_1$ and $v_2$.   

\begin{lm}\label{outsym}  For $N=3$, the largest eigenvalue of $P^{(2)}$ on the orthogonal complement of the symmetric sector is no greater than $0.735$. Thus, either 
$\widehat{\Delta}_{3,2} \geq 0.015$, or else the gap eigenfucntion is symmetric.
\end{lm}

\begin{proof}[Proof of Lemma~\ref{outsym}]
For later use, we begin the proof for $N\geq 3$, and specialize to $N=3$ later.
Let $f$ be given by \eqref{type2}, where we may assume that $\varphi$ is orthogonal to the constants. Then
$$\frac1N w_{N,2}(v_1)(1- K)\varphi(v_1)  -  \frac1N w_{N,2}(v_2)(1- K)\varphi(v_2)  = \lambda(\varphi(v_1) - \varphi(v_2)\ .$$
Multiplying by $\varphi(v_1)$ and integrating,
\begin{equation}\label{aus6}
\frac1N \int_{\ST} w_{N,2}(v_1)|(1- K)\varphi(v_1)|^2 =  \lambda \langle \varphi,(1-K)\varphi\rangle\ .
\end{equation}
By \eqref{ptws}, 
\begin{equation}\label{ptwsXX}
w_{N,2}(v) =  \frac{N}{N-1} - \frac{N}{(N-1)^2}|v_k|^2\ .
\end{equation}
Using \eqref{ptwsXX} in \eqref{aus6} yields
\begin{equation}\label{aus7}
\frac{1}{N-1}\langle \varphi,(1-K)^2\varphi\rangle -  \frac{1}{(N-1)^2}\langle (1-K)\varphi,|v|^2(1-K)\varphi\rangle =   \lambda \langle \varphi,(1-K)\varphi\rangle\ .
\end{equation}
Now write $\sqrt{1-K}\varphi = \psi + \zeta$ where $\psi$ is orthogonal to the constants, the three components of $v$ and $|v|^2$.   Then $\zeta$ is an 
eigenvector of $K$ with eigenvalue 
$-1/(N-1)$, and  hence
\begin{equation}\label{aus17}
\frac{1}{N-1}\langle \varphi,(1-K)^2\varphi\rangle = \frac{1}{N-1}\langle \psi,(1-K)\psi\rangle \
 \ ,
\end{equation}
and
\begin{eqnarray}
\langle (1-K)\varphi,|v|^2(1-K)\varphi\rangle 
&=&  \langle \sqrt{1-K}\psi, |v|^2 \sqrt{1-K}\psi\rangle + \frac{N-2}{(N-1)^2}\||v|\zeta\|_2^2  \nonumber\\
&-& 2 \||v|\sqrt{1-K}\psi\|_2   \frac{\sqrt{N-1}}{N-2}\||v|\zeta\|_2^2 \label{aus18}\\
&\geq& \left(1-\frac1t\right)\langle \sqrt{1-K}\psi, |v|^2 \sqrt{1-K}\psi\rangle + (1-t)\frac{N-2}{(N-1)^2}\||v|\zeta\|_2^2\ ,\nonumber
\end{eqnarray}
for all $t>0$, where we have used the arithmetic-geometric mean inequality.

At this point we specialize to $N=3$, and carry out some explicit computations that could be done for all $N\geq 3$, but are then more cumbersome. 

Write $\zeta = \sum_{j=1}^4 a_j \eta_j(v)$ where, as before, $\eta(j)v) = {\bf e}\cdot v$ for $j=1,2,3$, and where $\eta_4(v) = |v|^2 -1$. One readily computes that
\begin{equation}\label{aus20}
\||v|\zeta\|_2^2  = \sum_{j=1}^4 |a_j|^2 \||v|\eta_j\|_2^2
\end{equation}
and that 
${\displaystyle \int_{{\mathcal S}_3}|v_1|^4{\rm d}\sigma_3 = \frac 54}$ and ${\displaystyle \int_{{\mathcal S}_3}|v_1|^6{\rm d}\sigma_3 = \frac74}$.
From here  it follows that 
$$\||v|\eta_j\|_2^2 = \frac54 \quad{\rm for}\quad j=1,2,3\quad{\rm and}\quad   \||v|\eta_4\|_2^2  =   1\ . $$
Using this in \eqref{aus20} finally yields $\||v|\zeta\|_2^2 \geq \|\zeta\|_2^2$, and evidently 
$\langle \sqrt{1-K}\psi, |v|^2 \sqrt{1-K}\psi\rangle \leq 3\|\sqrt{1-K}\psi\|_2^2$.   Therefore, for $0 < t < 1$, we have from \eqref{aus18} that
$$\langle (1-K)\varphi,|v|^2(1-K)\varphi\rangle   \geq  3\left(1-\frac1\lambda\right) \|\sqrt{1-K}\psi\|_2^2 + (1-\lambda)\|\zeta\|_2^2\ .$$
Using this estimate in \eqref{aus7} yields
$$\frac12  \|\sqrt{1-K}\psi\|_2^2 + \frac34  \|\zeta\|_2^2 - 3\left(1-\frac1t\right) \|\sqrt{1-K}\psi\|_2^2 - (1-t)\|\zeta\|_2^2 \geq \lambda(\|\psi\|_2^2+\|\zeta\|_2^2)\ .$$

The second most negative eigenvalue of $K$ for $N=3$ is  $-\frac38$; see \cite[Section 8]{CGL}, where this eigenvalue is denoted $\kappa_{1,2}$. It follows that
$\|\sqrt{1-K}\psi\|_2^2 \leq \frac{11}{8}\|\psi\|_2^2$.  Therefore,
$$\left(\frac{1}{16} - 3 + \frac3t\right)\|\psi\|_2^2 + \left(\frac34 - 1 + t\right)\|\zeta\|_2^2 \geq \lambda (\|\psi\|_2^2 + \|\zeta\|_2^2)\ .$$
Choosing $t= 0.985$, we have that
$0.735(\|\psi\|_2^2 + \|\zeta\|_2^2) \geq \lambda (\|\psi\|_2^2 + \|\zeta\|_2^2)$.
\end{proof}

The remaining task is to bound the second largest eigenvalue of $P^{(2)}$  in the symmetric sector.
We begin considering general $N\geq 3$ and shall specialize to $N=3$ later. 

Let $f$ be given as in \eqref{type1}. Then  $P^{(2)}f = \lambda f$ becomes
\begin{equation}\label{right1}
\frac1N \sum_{k=1}^N w_{N,2}(v_k)( \varphi(v_k) +(N-1)K \varphi(v_k)) = \lambda  \sum_{j=1}^N \varphi(v_j)\ .
\end{equation}
By Theorem~\ref{gprop2}, 
\begin{equation}\label{right2}
 \frac1N w_{N,2}(v)( \varphi(v) +(N-1)K \varphi(v)) =\lambda  \varphi(v) \ .
\end{equation}
we have that 
Therefore, multiplying both sides of 
\eqref{right2} by $\varphi(v)$ and integrating, we obtain
\begin{equation}\label{right3}
 \frac1N \int_{\ST} \varphi(v_1) w_{N,2}(v_1)( \varphi(v_1) +(N-1)K \varphi(v_1)){\rm d}\sigma_N   =  \lambda  \|\varphi\|_2^2\ .
\end{equation}

By \eqref{ptws}, \eqref{right3} becomes
\begin{multline}\label{right4B}
\langle \varphi, K \varphi\rangle  - 
 \frac{1}{N-1} \int_{\ST} \varphi(v_1) |v_1|^2K \varphi(v_1)){\rm d}\sigma_N   =\\  \left(\lambda  - \frac{1}{N-1}\right) \|\varphi\|_2^2 + \frac{1}{(N-1)^2}
 \int_{\ST} |v_1|^2 \varphi^2(v_1)){\rm d}\sigma_N\ .
\end{multline}
Define an operator $M$ by
$M\phi(v) = |v|^2(1 + (N-1)K)\phi(v)$ .

Then \eqref{right4B} becomes
$$
\left(\lambda  - \frac{1}{N-1}\right) = \frac{ \langle \varphi, K \varphi\rangle - (N-1)^{-2}\langle \varphi, M \varphi\rangle}{\|\varphi\|_2^2}\ ,
$$
Thus $\lambda - (N-1)^{-1}$ can be computed by computing the supremum of the right hand side as $\varphi$ ranges over functions that are orthogonal to $1$, the three components of $v$ and $|v|^2$.  

Note that $M$ commutes with rotations so the different angular momentum sectors are mutually orthogonal, and can be considered separately. In each sector, 
by the usual recursion relations for orthogonal polynomials, the matrix representing $M$ in the eigenbasis of $K$ is tri-diagonal and explicitly computable, and the bounds proved in \cite[Section 8]{CGL} can be used to limit the number of angular momentum sectors that need to be considered. Hence one could obtain explicit bounds this way.

\end{appendices}



\begin{thebibliography}{99}

\bibitem{Cap04} P.~Caputo, \textit{Spectral gap inequalities in product spaces with conservation laws}, In
\textit{Stochastic Analysis on Large Interacting Systems} T. Funaki and H. Osada, eds.,
Math. Soc. Japan, Tokyo. (2004), 53Ð88.

\bibitem{Cap08} P.~Caputo, \textit{ On the spectral gap of the Kac walk and other binary collision processes},
ALEA Lat. Am. J. Probab. Math. Stat., {\bf 4}, (2008), 205-222.

\bibitem{CCC} {E.~A.~Carlen, J.~A.~Carrillo and M.~C.~ Carvalho:}
\textit{ Strong convergence towards homogeneous cooling states for dissipative Maxwell models},  Ann.  l'Institut H. Poincare (C) Non Linear Analysis
{\bf 26}, 5, 2009, 1675-1700

in \textit{Journes Equations aux derivees partielles}, Nantes, 5-9 Juin 2000.

\bibitem{CCL00} {E.~A.~Carlen, M.~C.~ Carvalho and M.~Loss:}
\textit{ Many body aspects of approach to equilibrium},
in \textit{Journes Equations aux derivees partielles}, Nantes, 5-9 Juin 2000.

\bibitem{CCL03}  E.~A.~Carlen, M.~C.~Carvalho and M.~Loss: \textit{ Determination of the spectral
gap for Kac's master equation and related stochastic evolution},
Acta Mathematica {\bf 191},  (2003) 1-54     

\bibitem{CCL14}  E.~A.~Carlen, M.~C.~Carvalho and M.~Loss: \textit{ Spectral gap for the Kac model with hard sphere collisions}. J. Funct. Anal. {\bf 266} (2014), no. 3, 1787-1832
   

\bibitem{CGL} E.~A.~Carlen, J.~Geronimo and M.~Loss:  \textit{Determination of the spectral gap in the Kac model for 
physical momentum and energy conserving collisions}, SIAM Jour. Mathematical Analysis {\bf 40},
no. 1, (2008) 327-364

\bibitem{Fell} W.~Feller, \textit{An introduction to Probability Theory and its Applications, Volume I},  Wiley and sons, New York, 1950

\bibitem{GF08} G.~Giroux and R.~Ferland, \textit{Global spectral gap for DirichletÐKac random motions}. J. Stat. Phys. {\bf 132}, (2008)
 561-567. 
 
\bibitem{GKS12}
 A.~Grigo, A.~Khanin and D.~ Sz\'asz, \textit{Mixing rates of particle systems with energy exchange}, Nonlinearity {\bf 25} (2012), 2349-2376.

\bibitem{J01} E.~Janvresse, \textit{ Spectral Gap for Kac's model of 
Boltzmann Equation}, Annals. of Prob., {\bf 29} (2001) 288-304.


\bibitem{K56} M.~Kac, \textit{Foundations of kinetic theory}, Proc. 3rd Berkeley
symp. Math. Stat. Prob., J. Neyman, ed. Univ. of California, vol 3, (1956)
171--197.

\bibitem{K59}\ M.~Kac \textit{ Probability and Related Topics in Physical Sciences},
Interscience Publ. LTD., London, New York (1959)

\bibitem{M66}
F.G.~Mehler
\textit{ \"Uber die {E}ntwicklung einer {F}unction von beliebig vielen
  {V}ariablen nach {L}aplaschen {F}unctionen h\"oherer {O}rdnungn.}
 Crelle's Journal {\bf 66}  (1866), 161--176.
 
 \bibitem{MM13} S.~Mischler and C.~Mouhot, \textit{Kac's Program in Kinetic Theory}, Invent. Math., {\bf 193} (2013) 1-147. 
 
 \bibitem{Sas13} M.~Sasada,  \textit{On the spectral gap of the Kac walk and other binary collision processes on $d$-dimensional lattice}, In 
 \textit{Symmetries, Integrable Systems and Representations},  K. Iohara et al., eds. Springer Proc. Math. Stat. {\bf 40},  543Ð560. 
 Springer, Heidelberg, 2013.
 
 \bibitem{Sas15} M.~Sasada,  \textit{Spectral gap for particle systems with degenerate energy exchange rates}, 
 Ann. Probab.
{\bf 43} (2015),1663-1711.

\bibitem{Szego} G.~Szeg\"o, \textit{Orthogonal Polynomials}, Vol. 23 of A.M.S. Colloquium Series Publications, A.M.S., Rovidence, 1967. 

\bibitem{Sznit} A.~S.~Sznitman, \textit{Topics in propagation of chaos}. In: \textit{\'Ecole dÕ\'Et\'e de Probabilit\'es de Saint-Flour XIXÑ1989}. Lecture Notes in Math., {\bf 1464}, pp. 165Ð251. Springer, Berlin, 1991.


\bibitem{V03} C.~Villani, \textit{Cercignani's Conjecture is Sometimes True and Always Almost True},
Comm. Math. Phys., 
{\bf 234}, (2003) 455--490.








\end{thebibliography}
\end{document}